\pdfoutput=1
\documentclass[10pt,pra,aps,superscriptaddress,nofootinbib,twocolumn,longbibliography]{revtex4-1}
\usepackage[T1]{fontenc}
\usepackage{amsmath,bbm}
\usepackage{amsthm}
\usepackage{amsmath}
\usepackage{latexsym}
\usepackage{amssymb}
\usepackage{float}
\usepackage{graphicx}           
\usepackage{color}
\usepackage{xcolor}
\usepackage{mathpazo}
\usepackage{comment}
\usepackage{enumitem}
\usepackage{multirow}
\usepackage{setspace}
\usepackage{graphicx}
\usepackage{caption}
\usepackage{subcaption}
\usepackage{booktabs}

\usepackage[colorlinks=true,linkcolor=blue,citecolor=magenta,urlcolor=blue]{hyperref}
\usepackage{svg}

\newcommand{\be}{\begin{equation}}
\newcommand{\ee}{\end{equation}}
\newcommand{\bea}{\begin{eqnarray}}
\newcommand{\eea}{\end{eqnarray}}

\def\squareforqed{\hbox{\rlap{$\sqcap$}$\sqcup$}}
\def\qed{\ifmmode\squareforqed\else{\unskip\nobreak\hfil
\penalty50\hskip1em\null\nobreak\hfil\squareforqed
\parfillskip=0pt\finalhyphendemerits=0\endgraf}\fi}
\def\endenv{\ifmmode\;\else{\unskip\nobreak\hfil
\penalty50\hskip1em\null\nobreak\hfil\;
\parfillskip=0pt\finalhyphendemerits=0\endgraf}\fi}

\newcommand{\tr}{\text{Tr}}

\newcommand{\ket}[1]{|#1\rangle}
\newcommand{\bra}[1]{\langle#1|}

\newcommand{\la}{\langle}
\newcommand{\ra}{\rangle}

\newcommand{\re}{\color{black}}  
\newcommand{\blk}{\color{black}}

\makeatletter
\newtheorem*{rep@theorem}{\rep@title}
\newcommand{\newreptheorem}[2]{%
\newenvironment{rep#1}[1]{%
 \def\rep@title{#2 \ref{##1}}%
 \begin{rep@theorem}}%
 {\end{rep@theorem}}}
\makeatother

\newtheorem{thm}{Theorem}
\newreptheorem{thm}{Theorem}
\newtheorem{lemma}{Lemma}
\newtheorem{definition}{Definition}

\begin{document}

\title{
Communication scenario enabling semi-device-independent robust self-testing of $n$-party Greenberger–Horne–Zeilinger basis measurements
}


\author{Barnik Bhaumik}
\affiliation{School of Physics, Indian Institute of Science Education and Research Thiruvananthapuram, Kerala 695551, India}
\author{Sagnik Ray}
\affiliation{School of Physics, Indian Institute of Science Education and Research Thiruvananthapuram, Kerala 695551, India}
\author{Debashis Saha}
\affiliation{Department of Physics, School of Basic Sciences, Indian Institute of Technology Bhubaneswar, Odisha 752050, India}


\begin{abstract}
Entangled basis measurements play a crucial role in distributing quantum entanglement between parties across a quantum network. In this work we adopt a semi-device-independent approach that enables the self-testing of $n$-qubit Greenberger–Horne–Zeilinger basis measurements without requiring shared entanglement between distant parties. Our method relies solely on input-output statistics from a communication scenario involving $n$ spatially separated senders, each receiving two bits of input, and a single receiver with no input. We analyze the robustness of the proposed self-testing protocol for $n\leq7$. Additionally, we introduce a protocol for robust self-testing of the three-outcome partial Bell basis measurement that is easily implementable in an optical setup.

\end{abstract}

\maketitle

\section{Introduction}

The quantum network represents a major milestone in the ongoing quantum revolution, offering capabilities that significantly surpass classical networks in terms of communication efficiency, security, and distributed information processing \cite{simon2017towards, kimble2008quantum,PhysRevLett.74.1259,PhysRevLett.83.432}. An essential requirement for realizing a quantum network is to certify entangled basis measurements, which serve as an essential component in distributing quantum entanglement.

Self-testing has emerged as a powerful technique for certifying quantum resources in a device-independent manner \cite{Supic2020selftestingof}. It allows one to infer the underlying quantum state, measurement, or processes solely from observed statistics, without requiring any trust in the internal functioning of the devices. Specifically, self-testing exploits extremal quantum correlations—typically those that maximally violate Bell inequalities—to uniquely determine the state and measurements up to local isometries. The foundational work by Mayers and Yao demonstrated that the maximal violation of the Clauser-Horne-Shimony-Holt (CHSH) inequality \cite{PhysRevLett.23.880} certifies, in a device-independent manner, that two parties share a singlet state up to local isometries \cite{Mayers:2003xif}. Since then, the self-testing of multipartite entangled quantum states has gained significant attention in Bell scenarios \cite{Supic2020selftestingof,Kan, PhysRevA.87.050102, C, b, bamps, sarkar, upi, Santos, Panwar, genuine, debashis, Kanie, Su, Vo, man, li}.
In contrast, self-testing of multipartite entangled basis measurements is subtler and inherently necessitates complex configurations with multiple independent sources. Previous works have demonstrated self-testing of entangled measurements in entanglement swapping scenarios within the multiparty network configurations \cite{Rabelo_2011,Renou-prl, Bancal-prl, Zhou-pra,sarkar2024universalselftest}.

Device-independent approaches, whether for certifying entangled states or measurements, rely on establishing entanglement between spatially separated parties---a task that remains experimentally demanding. To address this challenge, semi-device-independent approaches have been developed for self-testing of single quantum states or measurements \cite{supic, PhysRevA.98.022311, PhysRevResearch.3.033093, PhysRevA.106.L040402, PhysRevApplied.19.034038, PhysRevA.98.062307, Mate-pra-2019, PhysRevLett.122.250403, Saha2020sumofsquares, Das2022robustcertification, PhysRevLett.132.140201, baroni2025translatingbellnonlocalityprepareandmeasure, PhysRevLett.131.250802, Metger2021selftestingofsingle, PhysRevA.104.022212, PRXQuantum.3.030344}. These methods rely on minimal and general assumptions about the uncharacterized devices. This naturally leads to an important question: Can one self-test entangled basis measurements without requiring shared entanglement between distant parties?

To address this, we investigate the multiparty communication scenario comprising multiple senders and a single receiver. Each sender communicates unknown quantum systems to the receiver, who then performs an uncharacterized measurement. Our goal is to certify, up to local unitaries, that the receiver's measurement corresponds to the entangled basis comprised of Greenberger-Horne-Zeilinger (GHZ) states. The only assumption made is an upper bound on the dimension of the communication systems. There has been established evidence that measurements of entanglement bases are required to achieve certain quantum correlations that cannot be achieved using any product measurements or classical systems \cite{Vertesi-pra,Marcin-pra,expt-prl,chakraborty2024overcomingtraditionalnogotheorems,qfp}. This work marks a significant step forward by introducing a class of quantum communication tasks in which the optimal quantum performance self-tests the $n$-party GHZ basis measurement.

The paper is organized as follows. We begin in Sec. \ref{Self-testing in Communication Scenario} by describing the general communication scenario involving multiple spatially separated senders and a single receiver, and outline how the observed input-output statistics can be used for self-testing. We then discuss why not all extremal quantum correlations in such communication settings imply self-testing. 
The communication task used for self-testing is then introduced, and a specific case involving two senders and one receiver is explicitly demonstrated. Sec. \ref{self-testing GHZ measurments} is dedicated to proving the self-testing of the $n-$party GHZ measurement from the maximum value of the success metric of the proposed communication task. In Sec. \ref{Robust Self-testing of partial Bell basis measurements} we analyze the robustness of the self-testing protocol \re for $n\leq7$ \blk.
Motivated by experimental feasibility, we then modify the task to design a protocol capable of self-testing a three-outcome partial Bell basis measurement that can be readily implemented in an optical setup. We conclude in Sec. \ref{conclusion} with a discussion of the broader implications, limitations, and potential directions for future research.

\section{Self-testing in a Communication Scenario}\label{Self-testing in Communication Scenario}
 
We consider a quantum communication scenario that involves $n$ spatially separated senders $\{A^{(j)}\}_{j=1}^n$ and a single receiver $R$. The $j$-th sender receives an input $y_j$ and encodes it by preparing a $d$-dimensional quantum state $\rho^{(j)}_{y_j}$, which is then transmitted to the receiver. In this work, we consider a scenario in which the receiver receives an input $k$. The receiver performs a joint measurement on the composite quantum state $\bigotimes_j\rho^{(j)}_{y_j}$ to produce an output $s$. The resulting input-output statistics are given by the conditional probabilities, 
\begin{equation}
p(s|\vec{y},k) = \tr\Big(\bigotimes_{j}\rho^{(j)}_{y_j}\mathcal{M}_s^k\Big),
\end{equation}
where $\vec{y} = (y_1,\cdots,y_n)$ denotes the set of inputs received by the senders, and $\{\mathcal{M}_s^k\}_s$ is a set of positive semi-definite operators describing the measurement by the receiver corresponding to the $k$-th input, satisfying the completeness relation $\sum_s \mathcal{M}_s^k = \mathbbm{1}$, for each $k$. 

To characterize the quantum behavior in this communication scenario, one typically considers a success metric that is a linear combination of such probabilities,
\begin{equation}\label{success metric}
 \mathcal{S}=\sum_{\vec{y},s,k}\alpha_{\vec{y},s,k}p(s|\vec{y},k)  ,
\end{equation}
where $\alpha_{\vec{y},s,k}\in\mathbb{R}.$ For a fixed system dimension $d$, quantum systems can attain certain values of $\mathcal{S}$ that are impossible to replicate using classical systems of the same dimension, even when allowing shared classical randomness between parties \cite{Vertesi-pra,Marcin-pra,expt-prl,chakraborty2024overcomingtraditionalnogotheorems,qfp}. This phenomenon is referred to as the quantum advantage over classical communication. 

However, our primary aim is to demonstrate that the optimal value of $\mathcal{S}$ for a certain communication task (as we will see, it has no input on the receiver's end) can only arise from a specific entangled basis measurement performed by the receiver, up to local unitary transformations.
\begin{definition}
     Self-testing of a reference entangled basis measurement $\{\ket{\xi_{s}}\}_s$ (where $\ket{\xi_s}$ are entangled states) from the optimal value of $\mathcal{S}$ asserts the existence of a set of unitaries $\{U_j\}_j$ such that the unknown measurement $\{\mathcal{M}_s\}_s$ performed in the receiver satisfies:
     \begin{equation}
         \left(\otimes_j U_j\right) \mathcal{M}_s   \left(\otimes_j U_j\right)^\dagger = \ket{\xi_{s}}\!\bra{\xi_s}\ \forall s.
     \end{equation}
\end{definition}
This definition captures the essence of self-testing that allows one to characterize an unknown quantum measurement as an ideal target measurement, without requiring detailed knowledge of the internal workings of the measurement device. Although the main focus of this work is the self-testing of the measurement at the receiver's end, our results also encompass the self-testing of the quantum states prepared by the senders. Since we assume an upper bound on the dimension of the communicated systems, the self-testing holds in the semi-device-independent regime. It is also important to note that one can efficiently estimate a lower bound on the fidelity between the implemented (uncharacterized) measurement and the ideal target measurement. This provides a quantitative means to assess the robustness of the self-testing protocol.

\subsection{Optimal quantum advantage does not always imply self-testing} \label{no_advantage}
However, before diving into a detailed proof of the self-testing protocol, let us first consider an instance where an optimal quantum advantage does \textit{not} implicitly suggest that we can certify entangling measurements or, in other words, self-test them. This example has been provided here to emphasize the novelty of the communication task discussed in the next section, which enables us to self-test GHZ basis measurements. Consider a scenario where there are two spatially separated senders $A^{(1)}$ and $A^{(2)}$ and each gets an input $y_1\in\{1,2,3\}$ and $y_2\in \{1,2,3\}$ respectively. Depending on their input, they can send a qubit state $\{\rho_{y_j}^{(j)}\}_{y_j=1}^{3}\in\mathbb{C}^2 \textrm{ }\forall j=1,2$ to the receiver. The receiver can perform a two-outcome measurement $\{\mathcal{M}_0,\mathcal{M}_1\}$ on the composite state and generate probability statistics as
\begin{equation}
    p(s|y_1,y_2)= \tr\left[\left(\rho_{y_1}^{(1)} \otimes \rho_{y_2}^{(2)}\right) \mathcal{M}_s\right], \quad s=0,1.
\end{equation}
Let us consider the success metric, 
\begin{equation} \label{S_no_self_test}
\begin{split}
\mathcal{S}=&-2[p(0|1,1)-p(0|1,3)+p(0|2,1)]\\ &+p(0|2,2)-p(0|2,3)+p(0|3,2)-p(0|3,3),
\end{split}
\end{equation}
which corresponds to a facet inequality of the polytope characterizing the set of classical correlations in this scenario \cite{Vertesi-pra}. We observe that both unentangled and entangled measurements can surpass the classical bound associated with this metric. Notably, while entangled measurements generally provide an advantage over classical strategies, there exist instances where unentangled measurements achieve the same level of success as their entangled counterparts. This suggests that the advantage of entanglement, while present, may not always manifest in the maximum value for this specific metric and in turn cannot be self-tested.

The maximum possible value of $\mathcal{S}\approx2.8284$, and there exist two different measurements, one entangling, $\mathcal{M}_0^{\textrm{ent}}$, and another nonentangling, $\mathcal{M}_0^{\textrm{nonent}}$ that result in $\mathcal{S}$ achieving the aforementioned value. The entangling measurement can be expressed as
\begin{equation}
    \mathcal{M}_0^{\textrm{ent}}=\ket{\Psi}\!\bra{\Psi}+\ket{\Psi^\perp}\!\bra{\Psi^\perp},
\end{equation}
where $\ket{\Psi}=\lambda_1\ket{00}+\lambda_2\ket{11}$ with $(\lambda_1,\lambda_2)\approx(0.9413,0.3375)$ and  with $\ket{\Psi^\perp}=\tilde{\lambda}_1\ket{\vec{n}}\ket{\vec{m}}+\tilde{\lambda}_2\ket{-\vec{n}}\ket{-\vec{m}}$ with $(\tilde{\lambda}_1,\tilde{\lambda}_2)\approx(0.9240,0.3357)$. The states $\ket{\pm\vec{n}}$ and $\ket{\pm\vec{m}}$ are the eigenstates corresponding to $\pm1$ eigenvalues of $\vec{n}.\vec{\sigma}$ and $\vec{m}.\vec{\sigma}$ respectively where $\vec{n}=(\sin(\theta)\cos(\phi),\sin(\theta)\sin(\phi),\cos(\theta))$ and $\vec{m}=(\sin(\theta')\cos(\phi'),\sin(\theta')\sin(\phi'),\cos(\theta'))$ with $\theta\approx179.61^o, \phi\approx354.23^o$ and $\theta'\approx48.93^o, \phi'\approx116.69^o$. One can check the entanglement by applying positive partial transpose
\cite{PhysRevLett.77.1413}. The optimal messages in this case are $\rho_{y_1}^{(1)}=\ket{\psi_{y_1}}\!\bra{\psi_{y_1}}$ and $\rho_{y_2}^{(2)}=\ket{\overline{\psi}_{y_2}}\!\bra{\overline{\psi}_{y_2}}\quad \forall y_1,y_2\in\{1,2,3\}$. The polar and azimuthal angles of every such $\ket{\psi_{y_1}}=\cos(\theta_{y_1}/2)\ket{0}+e^{\iota\phi_{y_1}}\sin(\theta_{y_1}/2)\ket{1}$ and $\ket{\overline{\psi}_{y_2}}=\cos(\overline{\theta}_{y_2}/2)\ket{0}+e^{\iota\overline{\phi}_{y_2}}\sin(\overline{\theta}_{y_2}/2)\ket{1}$ are listed in Table \ref{tab:merged_mixed_headers}.

\vspace{0.5cm}

\begin{table}[h!]
\centering
\begin{tabular}{lcc|lcc}
  \toprule
  $A^{(1)}$ & $\theta_i$ & $\phi_i$ & $A^{(2)}$ & $\overline{\theta}_i$ & $\overline{\phi}_i$ \\
  \midrule
  $\psi_1$           & $118.05^\circ$ & $0^\circ$ & $\overline{\psi}_1$ & $151.45^\circ$ & $287.40^\circ$ \\
  $\psi_2$           & $125.58^\circ$  & $243.78^\circ$ & $\overline{\psi}_2$ & $69.76^\circ$  & $116.28^\circ$ \\
  $\psi_3$           & $125.73^\circ$  & $244.07^\circ$ & $\overline{\psi}_3$ & $65.49^\circ$  & $296.09^\circ$ \\
  \bottomrule
\end{tabular}
\caption{Messages sent by $A^{(1)}$ and $A^{(2)}$ when measurement is entangling.}
\label{tab:merged_mixed_headers}
\end{table}

The same value of the success metric is also achieved by a separable measurement $\mathcal{M}_0^{\textrm{nonent}}$
\begin{equation}
    \mathcal{M}_0^{\textrm{nonent}}=\ket{0}\!\bra{0}\otimes\ket{u}\!\bra{u}+\ket{1}\!\bra{1}\otimes\ket{0}\!\bra{0},
\end{equation}
where \re $\ket{u}=\cos(\theta/2)\ket{0}+e^{\iota\phi}\sin(\theta/2)\ket{1}$ with $\theta\approx46.58^\circ$ and $\phi\approx208.09^\circ$. The messages sent by the senders are listed in Table \ref{tab:merged_mixed_headers1}. \blk

\vspace{0.5cm}

\begin{table}[h!]
\centering
\begin{tabular}{lcc|lcc}
  \toprule
  $A^{(1)}$ & $\theta_i$ & $\phi_i$ & $A^{(2)}$ & $\overline{\theta}_i$ & $\overline{\phi}_i$ \\
  \midrule
  $\psi_1$           & \re$55.95^\circ$\blk & \re$0^\circ$ \blk & $\overline{\psi}_1$ & \re$127.29^\circ$\blk & \re$59.56^\circ$ \blk\\
  $\psi_2$           & \re$123.99^\circ$ \blk & \re$208.09^\circ$\blk & $\overline{\psi}_2$ & \re$21.71^\circ$\blk  & \re $28.09^\circ$ \blk\\
  $\psi_3$           & \re$123.99^\circ$ \blk & \re$208.09^\circ$\blk & $\overline{\psi}_3$ & \re$113.28^\circ$ \blk & \re $208.09^\circ$ \blk \\
  \bottomrule
\end{tabular}
\caption{Messages sent by $A^{(1)}$ and $A^{(2)}$ when measurement is separable.}
\label{tab:merged_mixed_headers1}
\end{table}

Thus, as claimed, since local unitaries cannot connect the entangling and nonentangling measurements, it implies that the entangling measurement cannot be self-tested using the optimal success metric of this communication scenario.
\section{Self-Testing of GHZ measurements} \label{self-testing GHZ measurments}

Before delving into self-testing of GHZ basis measurements, we first introduce the underlying communication task and define the metric used to evaluate its success.
\vspace{0.5cm}

\subsection{Communication task} \label{comm task}
The communication scenario involves $n$ spatially separated senders, denoted by $\{A^{(1)},A^{(2)},\cdots,A^{(n)}\}$, and a receiver $R$. Each sender $A^{(j)}$ receives two inputs $x_j$ and $a_j$, where $x_j$ and $a_j$ $\in \{0,1\}$.  The index $j$ runs over the set of positive integers, i.e. $j\in$ $\{1,2,\cdots,n\}$. Based on the combination of inputs, they prepare and send messages $m_d(x_j,a_j)$ (classical or quantum), where $d=2$, to the receiver, who can output an $n$-bit binary number $s=s_ns_{n-1}\cdots s_2s_1$ where each bit $s_k$ is either $0$ or $1$ $\forall k \in \{1,\cdots,n\}$, as shown in FIG.~\ref{fig:enter-label}. Thus, in total, there are $2^n$ possible outputs.  After having repeated the task for several rounds, the receiver gathers correlation statistics, $p(s\ |\ \vec{a} ;\vec
{x})$ where $\vec{a}=(a_1,a_2,\cdots,a_n)$ and $\vec{x}=(x_1,x_2,\cdots,x_n)$ .
\begin{figure}[h] 
    \centering
    \includegraphics[width=0.5\textwidth]{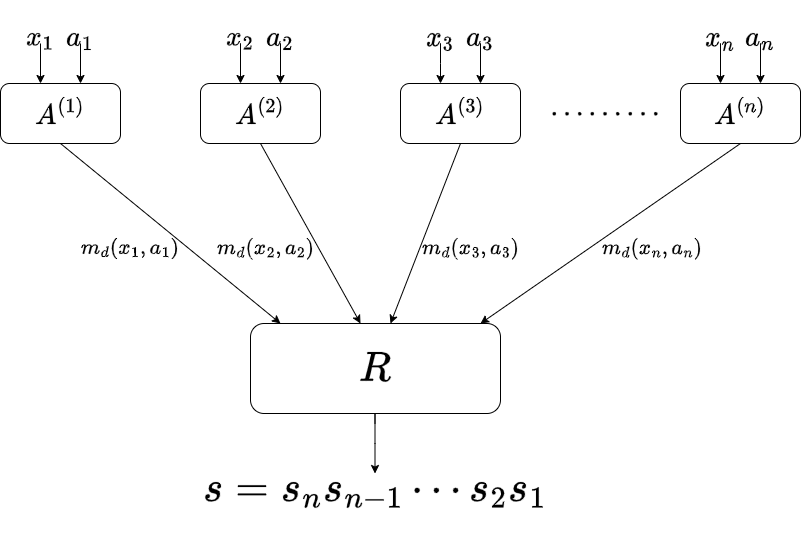}
    \caption{Schematic diagram of a multi-party communication scenario in which multiple senders $\{A^{(j)}\}_j$, each transmits a message $m_d(x_j,a_j)$, with $d=2$ in our case, determined by the respective  inputs $(x_j,a_j)$ where $j\in\{1,2,\cdots,n\}$. These messages are sent to the receiver, who produces an $n$-bit binary output.}
    \label{fig:enter-label}
\end{figure}

Once all correlation statistics are collected, the receiver calculates the success metric, which evaluates the overall performance of the communication protocol. The success metric is defined as
\be \label{SGen}
\mathcal{S} = \frac{1}{2^{n} (n-1) 2\sqrt{2}} \sum_s W_{s} \quad \textrm{where,}
\ee 

\begin{widetext}
\begin{equation} \label{SGen1}
\begin{split}
W_{s}=(n-1) (-1)^{s_1}&\sum_{a_1,\cdots,a_n=0,1}\left(-1\right)^{\bigoplus_{j=1}^n a_j}\Big[p(s\ |\ \vec{a};x_1=\cdots=x_n=0)+p(s\ |\ \vec{a};x_1=1,x_2=\cdots=x_n=0)\Big] +\\  &\sum_{j=2}^n (-1)^{s_j} \sum_{a_1,a_j=0,1}(-1)^{a_1\oplus a_j}\Big[p(s\ |\ a_1,a_j;x_1=0,x_j=1)-p(s\ |\ a_1,a_j;x_1=x_j=1) \Big].
\end{split}
\end{equation}
\end{widetext}
The success metric $\mathcal{S}$ is normalized so that its maximum quantum value is ensured to 1. (\re Note that any inputs $a_{j'}$ and $x_{j'}$ for $j'\neq1,j$ are irrelevant for the terms, $p(s\ |\ a_1,a_j;x_1=0,x_j=1)$and $p(s\ |\ a_1,a_j;x_1=x_j=1)$, where $j\in\{2,\cdots,n\}$)\blk. 

As an explicit example, let us consider a two sender-one receiver scenario where the two senders are spatially separated and sends their messages to the receiver. In this case, Eq.\eqref{SGen} can be reduced to, 
\be \label{SGen1}
\mathcal{S} = \frac{1}{8\sqrt{2}} \sum_s W_{s},
\ee
where $s=s_2s_1$ can take four values,$\{00,01,10,11\}$ and $W_s$ is defined as
\begin{widetext} \label{chsh bell 2 sender}
    \begin{equation}
    \begin{split}
        W_s=\sum_{a_1,a_2=0,1}&(-1)^{a_1 \oplus a_2}\Big[(-1)^{s_1}[p(s|a_1,a_2;x_1=x_2=0)+p(s|a_1,a_2;x_1=1,x_2=0)]\\
       &+(-1)^{s_2}[p(s|a_1,a_2;x_1=0,x_2=1)-p(s|a_1,a_2;x_1=x_2=1)]\Big].
       \end{split}
    \end{equation}
\end{widetext}

Since the dimension of the communicated systems is restricted to $2$, when the senders use quantum states to encode their inputs, they send qubit states. For example, say the sender $A^{(j)}$ receives the inputs $\{a_j,x_j\}$, then the sender will send the state (mixed or pure) $\rho^{(j)}_{a_j|x_j}\in\mathbb{C}^2$. In the case of an $n$-party communication protocol, the measurement set consists of $2^n$ distinct elements, ensuring a complete set of possible measurement outcomes as an $n$-bit binary number $s$. Thus, the set $\{\mathcal{M}_s\}_{s}$ forms a valid positive operator valued measurement (POVM) and the term $W_s$ in (\ref{SGen}) can be written succinctly as $W_s=\tr(\mathcal{M}_s\mathcal{W}_s)$ where $\mathcal{W}_s$ is defined as
\begin{widetext}
\begin{equation} \label{w_i_def}
    \mathcal{W}_s=(n-1)(-1)^{s_1}\left(\mathcal{A}_0^{(1)}+\mathcal{A}_1^{(1)}\right)\otimes\bigotimes_{j=2}^n\mathcal{A}_0^{(j)}+ \sum_{j=2}^n (-1)^{s_{j}}\left(\mathcal{A}_0^{(1)}-\mathcal{A}_1^{(1)}\right)\otimes\mathcal{A}_1^{(j)},
\end{equation}
 with $\mathcal{A}^{(j)}_{x_j}$ defined as, $\mathcal{A}^{(j)}_{x_j}=\rho^{(j)}_{0|x_j}-\rho^{(j)}_{1|x_j}$.
    
\end{widetext}
In a communication task involving two senders and a receiver, Eq. \eqref{SGen} simplifies to
\begin{equation}\label{CHSH operator}
    \mathcal{S}=\frac{1}{8\sqrt{2}}\sum_{s}{W}_{s},
\end{equation}
where the operator $\mathcal{W}_s$ in \eqref{w_i_def} reduces to
\begin{equation} \label{chsh OPERATORS}
\begin{split}
\mathcal{W}_s=&(-1)^{s_{1}}(\mathcal{A}_{0}^{(1)}+ \mathcal{A}_1^{(1)})\otimes\mathcal{A}_0^{(2)}\\&+(-1)^{s_{2}}(\mathcal{A}_0^{(1)}-\mathcal{A}_1^{(1)})\otimes\mathcal{A}_1^{(2)}.
 \end{split}
\end{equation}  

\re Each $\mathcal{W}_s$ in \eqref{CHSH operator} resembles a CHSH-Bell operator; however, in our case, each term involves a combination of density operators rather than measurement operators. Moreover, the general expression for  $\mathcal{W}_s$ defined in \eqref{w_i_def} coincides with Bell inequalities that is maximally violated by qubit graph states as discussed in \cite{Baccari_2020}.  \blk

Having described the communication scenario, we now proceed with presenting the self-testing protocol.
\vspace{-0.01cm}

\subsection{Ideal self-testing of GHZ measurements}
\begin{thm}\label{th1}
    The optimal quantum value of $\mathcal{S}$ in \eqref{SGen} is $1$. If this value is achieved from an unknown set of qubit states $\{\rho^{(j)}_{a_j|x_j}\}_{a_j,x_j}$ and an unknown measurement $\{\mathcal{M}_s\}_s$ on the $n$-qubit system, then there exists a set of qubit unitaries $\{U_j\}_j$ such that
\begin{equation}\label{first set}
        \begin{split}
            &U_{1}  \rho^{(1)}_{0|0}  U_{1}^\dag =|\beta_{+}\rangle\!\langle \beta_{+}|, \ \  U_{1}  \rho^{(1)}_{0|1}  U_{1}^\dag =|\alpha_{+}\rangle\!\langle \alpha_{+}|, \\
          &U_{1}  \rho^{(1)}_{1|0} U_{1}^\dag =|\beta_{-}\rangle\!\langle \beta_{-}|,
            \ \ U_{1}  \rho^{(1)}_{1|1}  U_{1}^\dag =|\alpha_{-}\rangle\!\langle \alpha_{-}|,
            \end{split}
        \end{equation}
        for $j\neq 1$
         \begin{equation}\label{second set}
         \begin{split}
        &U_{j}  \rho^{(j)}_{0|0}  U_{j}^\dag =|0\rangle\!\langle 0|, \ \ U_{j}  \rho^{(j)}_{0|1}  U_{j}^\dag =|+\rangle\!\langle +|, \\
            & U_{j}  \rho^{(j)}_{1|0} U_{j}^\dag =|1\rangle\!\langle 1|, \ \  U_{j}  \rho^{(j)}_{1|1}  U_{j}^\dag =|-\rangle\!\langle -|,
            \end{split}
        \end{equation}
 and 
\begin{equation}
\left(\otimes_jU_j\right)\mathcal{M}_s\left(\otimes_j U_j\right)^{\dagger}=\ket{\xi_s}\!\bra{\xi_s},
\end{equation}
where
        \begin{equation}
        \begin{split}
            &|\beta_{+}\rangle=\cos \frac{\pi}{8} |0\rangle +\sin \frac{\pi}{8} |1\rangle, \ |\beta_{-}\rangle=\sin \frac{\pi}{8} |0\rangle - \cos \frac{\pi}{8} |1\rangle, \\
            &|\alpha_{+}\rangle=\sin \frac{\pi}{8} |0\rangle + \cos \frac{\pi}{8} |1\rangle, \ |\alpha_{-}\rangle=\cos \frac{\pi}{8} |0\rangle - \sin \frac{\pi}{8} |1\rangle , \nonumber 
        \end{split}
        \end{equation}
and $\ket{\xi_s}=\frac{1}{\sqrt{2}}\Big(\ket{0 \textrm{ } s_2 \textrm{ }\cdots \textrm{ } s_n}+(-1)^{s_1}\ket{1 \textrm{ }\overline{s_2} \textrm{ }\cdots \textrm{ }\overline{s_n}}\Big)$ where $\overline{s}_j$ is the compliment of $s_j\in \{0,1\}$. 
\end{thm}
\vspace{-0.43 cm}
\begin{proof} 
The proof begins by establishing that when $\mathcal{S}=1$, that is when $\tr(\mathcal{M}_s\mathcal{W}_s)=2\sqrt{2}(n-1)$ $\forall s$, the set of unknown qubit messages $\{\rho^{(j)}_{a_j|x_j}\}_{a_j,x_j}$  corresponding to each input $x_j$ must form a collection of pure and mutually orthogonal states. We then demonstrate that the POVM elements $\{\mathcal{M}_s\}_s$ must have unit trace, and that this condition, together with orthogonality, implies the communicated states must be as presented in \eqref{first set} and \eqref{second set}. Finally, we show that the corresponding measurement operators must be maximally entangled.

In the context of self-testing, the shifted $\mathcal{W}_s$ operator with a sum-of-squares (SOS) decomposition is a useful technique to certify quantum states. The shifted $\mathcal{W}_s$ operator involves modifying the $\mathcal{W}_s$ operator in such a way that it becomes a positive semi-definite operator, allowing us to decompose it into a sum of squares of operators. We shift the   $\mathcal{W}_s$ operator by subtracting it from its maximum possible eigenvalue $\beta_Q=2\sqrt{2}(n-1)$ times the identity operator $\mathbbm{1}$ to make it positive semi-definite (see Appendix \ref{sos decomposition} for the proof of $\beta_Q=2\sqrt{2}(n-1)$)
\begin{equation}\label{shiifted bell operator}
\beta_Q\mathbbm{1}-\mathcal{W}_s=\sum_{k}\mathcal{O}_{sk}^\dag \mathcal{O}_{sk},
\end{equation}
where each $\mathcal{O}_{sk}$ is an operator constructed from the messages sent by the senders. The choice of terms make it evident that the entire operator is non-negative. The SOS decomposition consists of three operators $\{\mathcal{O}_{sk}\}_{k=1}^3$, where
\begin{equation} \label{sum square decomp}
    \begin{split}
        &\mathcal{O}_{s1}^\dag\mathcal{O}_{s1}=\frac{n-1}{\sqrt{2}}(\mathbbm{1}-\mathcal{P}_{1,s})^{2},\\
        &\mathcal{O}_{s2}^\dag\mathcal{O}_{s2}=\frac{1}{\sqrt{2}}\sum_{j=2}^{n}(\mathbbm{1}-\mathcal{P}_{j,s})^2,\\
        &\mathcal{O}_{s3}^\dag\mathcal{O}_{s3}=\sqrt{2}(n-1)(\mathbbm{1}-\frac{1}{2(n-1)}((n-1)\mathcal{P}_{1,s}^2+\sum_{j=2}^n\mathcal{P}_{j,s}^2)),
    \end{split}
\end{equation}
and $\mathcal{P}_{1,s}$, $\mathcal{P}_{j,s}$ are defined as
\begin{equation}
\begin{split}\label{SOS 1}
        \mathcal{P}_{1,s}&=(-1)^{s_1}\frac{1}{\sqrt{2}}\left(\mathcal{A}_0^{(1)}+\mathcal{A}_1^{(1)}\right)\otimes \bigotimes_{j=2}^n\mathcal{A}_0^{(j)},\\
        \mathcal{P}_{j,s}&=(-1)^{s_j}\frac{1}{\sqrt{2}}\left(\mathcal{A}_0^{(1)}-\mathcal{A}_1^{(1)}\right)\otimes\mathcal{A}_1^{(j)} \quad \text{  $\forall$ $j\in \{2,\cdots,n\}$}.
\end{split}
\end{equation}
Since we have restricted the dimension to 2, we eliminate the need for any local isometries, allowing us to directly verify the state without additional transformations. We know that the operators $\mathcal{A}^{(j)}_{x_j}$ as present in \eqref{w_i_def} are the difference between two states for a fixed $x_j$ but different values of $a_j$ which correspond to the messages sent by the sender $A^{(j)}$. It is proved in Appendix \ref{sos decomposition} that to have $\lVert \mathcal{W}_s \rVert = 2\sqrt{2}(n-1)$, all the messages sent by senders must be pure states, and for a given $x_j$ the messages corresponding to $a_j=0$ and $1$ must be orthogonal i.e.,
    \begin{equation}
        \langle \psi^{(j)}_{0|x_j}|\psi^{(j)}_{1|x_j} \rangle = 0 \quad \forall j\in \{1,2,\cdots,n\},
    \end{equation}
where $\{\ket{\psi^{(j)}_{a_j|x_j}}\}_{a_j,x_j}$ are the corresponding pure states sent by the senders to the receiver.
    
Since the measurement operators $\mathcal{M}_s$ are positive semi-definite they have a spectral decomposition of the form
\begin{equation}
\label{spec decmp}
\mathcal{M}_s = \sum_{s'} \lambda_{s,s'} |\xi_{s,s'}\rangle\! \langle \xi_{s,s'}|,
\end{equation}
where  $\lambda_{s,s'} >0$ are the eigenvalues, and $|\xi_{s,s'}\rangle$ are the corresponding eigenstates. \re  Here $s'$ represents the index for the $2^n$ basis that is possible for each $s$.\blk As stated before, when self-tested, the expectation value of the  operator equation \eqref{w_i_def} when measured with respect to such a measurement operator $\mathcal{M}_s$ must be equal to $2\sqrt{2}(n-1)$, i.e.,
\begin{equation}
    \mathrm{Tr}(\mathcal{M}_s \mathcal{W}_s)=2\sqrt{2}(n-1).
\end{equation}
This can be restated using (\ref{spec decmp}) as
\begin{equation} \label{measure_1}
\sum_{s'}\lambda_{s,s'}\langle \xi_{s,s'}|\mathcal{W}_s| \xi_{s,s'} \rangle=2\sqrt{2}(n-1).
\end{equation}
We know the maximum eigenvalue of the operator $\mathcal{W}_s$ is $2\sqrt{2}(n-1)$, i.e.,
\begin{equation} \label{measure_2}
    \langle \xi_{s,s'}|\mathcal{W}_s| \xi_{s,s'} \rangle\leq 2\sqrt{2}(n-1),
\end{equation}
which implies that
\begin{equation}
    \sum_{s'}\lambda_{s,s'}\langle \xi_{s,s'}|\mathcal{W}_s| \xi_{s,s'} \rangle\leq 2\sqrt{2}(n-1)\sum_{s'} \lambda_{s,s'}.
\end{equation}
Combining \eqref{measure_1} and \eqref{measure_2} we can infer that
\begin{equation} \label{tr_mi_geq1}
    \begin{split}
    &\sum_{s'} \lambda_{s,s'}\geq 1 \textrm{ } \textrm{ or,}\\
    &\tr(\mathcal{M}_s)\geq 1.
    \end{split}
\end{equation}
In the case of POVMs, the sum of all measurement operators $\mathcal{M}_s$ equals the identity operator $\mathbbm{1}$. For a system with $n$ number of senders, the combined Hilbert space has a dimension of $2^n$. Thus, the sum of traces of the POVM elements is 
\begin{equation}\label{eq28}
   \sum_{s} \mathrm{Tr}(\mathcal{M}_s)=2^n.
\end{equation}
For Eq. \eqref{eq28} to be consistent with \eqref{tr_mi_geq1}, we must necessarily have
\begin{equation}
\begin{split}
    &\mathrm{Tr}(\mathcal{M}_s)=1 \textrm{ or, alternatively,}\\
    &\sum_{s'} \lambda_{s,s'}=1.
    \end{split}
\end{equation}
Tracing out \eqref{shiifted bell operator} by multiplying it with $\mathcal{M}_s$ and knowing that $\tr(\mathcal{M}_s)=1$, we can say that $\tr({\mathcal{M}_s \sum_{k}\mathcal{O}_{sk}^\dag \mathcal{O}_{sk}})=0$, and subsequently, $ \bra{\xi_{s,s'}}\mathcal{O}_{sk}^\dag \mathcal{O}_{sk}\ket{\xi_{s,s'}}=0, \quad \forall s,s',k$ . This, in turn, implies that if the optimal quantum of $\mathcal{S}$ is obtained, then for every $s,s',k,$ 
\begin{equation}\label{eq30}
\mathcal{O}_{sk}\ket{\xi_{s,s'}}=0.
\end{equation}
Substituting $\mathcal{O}_{s1}^{\dag}\mathcal{O}_{s1}$ and $\mathcal{O}_{s2}^{\dag}\mathcal{O}_{s2}$ from \eqref{sum square decomp} and \eqref{SOS 1} into the relation \eqref{eq30}, we obtain
\begin{equation}\label{31}
\begin{split}
    &(-1)^{s_1}[\frac{1}{\sqrt2}(\mathcal{A}_0^{(1)}+\mathcal{A}_1^{(1)})\otimes \bigotimes_{j=2}^{n}\mathcal{A}_0^{(j)}]\ket{\xi_{s,s'}}=\ket{\xi_{s,s'}},\\
    &(-1)^{s_j}[\frac{1}{\sqrt2}(\mathcal{A}_0^{(1)}-\mathcal{A}_1^{(1)})\otimes\mathcal{A}_{1}^{(j)}]\ket{\xi_{s,s'}}=\ket{\xi_{s,s'}}.
\end{split}
\end{equation}
Considering the fact that $\left(\mathcal{A}_{x_j}^{(j)}\right)^2=\mathbbm{1}_2$ and for all $j\in\{1,2,\cdots,n\}$ and $ x_j \in \{0,1\}$ , the Eqs. \eqref{31} are equivalent to 
\begin{equation}\label{action of X and Z}
\begin{split}
    &\mathcal{X}\otimes\mathbbm{1}_{2^{n-1}}\ket{\xi_{s,s'}}=(-1)^{s_1}\mathbbm{1}_2\otimes\bigotimes_{j=2}^{n}\mathcal{A}_0^{(j)}\ket{\xi_{s,s'}},\\
    & \mathcal{Z}\otimes \mathbbm{1}_{2^{n-1}}\ket{\xi_{s,s'}}=(-1)^{s_j}\mathbbm{1}_2\otimes\mathcal{A}_{1}^{(j)}\otimes\mathbbm{1}_{2^{n-2}}\ket{\xi_{s,s'}},
\end{split}
\end{equation}

respectively, where $\mathcal{X}$ and $\mathcal{Z}$ are defined as
\begin{equation}\begin{split}\label{renaming}
    &\mathcal{X}=\frac{1}{\sqrt{2}}(\mathcal{A}_0^{(1)}+\mathcal{A}_1^{(1)}) \quad \textrm{and,}\\
    &\mathcal{Z}=\frac{1}{\sqrt{2}}(\mathcal{A}_0^{(1)}-\mathcal{A}_1^{(1)}).
\end{split}
\end{equation}
We need to show that $\{\mathcal{A}^{(j)}_0,\mathcal{A}^{(j)}_1\}=0$ $\forall j \in \{1,2,\cdots,n\}$. From \eqref{renaming} it is clear that $\{\mathcal{X},\mathcal{Z}\}=0$. Furthermore, Eqs. \eqref{action of X and Z} imply that
\begin{equation}
    \begin{split}
        &\bra{\xi_{s,s'}}\mathcal{X}^2\otimes\mathbbm{1}_{2^{n-1}}\ket{\xi_{s,s'}}=1 \quad \textrm{and,}\\
        &\bra{\xi_{s,s'}}\mathcal{Z}^2\otimes\mathbbm{1}_{2^{n-1}}\ket{\xi_{s,s'}}=1,
    \end{split}
\end{equation}
which results in $\mathcal{X}^2=\mathcal{Z}^2=\mathbbm{1}_2$. Thus, there must exist a single-qubit unitary $U_1$ such that $U_1\mathcal{X}U_1^\dag=\sigma_X$ and $U_1\mathcal{Z}U_1^\dag=\sigma_Z$, and hence, $\{\mathcal{A}^{(1)}_0,\mathcal{A}^{(1)}_1\}=0$.
For $j\in\{2,3,\cdots n\}$, consider the action of $\{\mathcal{X},\mathcal{Z}\}$ on $\ket{\xi_{s,s'}}$, using \eqref{renaming}
\begin{equation}
    (\mathcal{X}\mathcal{Z}+ \mathcal{Z}\mathcal{X})\otimes \mathbbm{1}_{2^{n-1}}\ket{\xi_{s,s'}}=0.
\end{equation}
This equation, along with \eqref{action of X and Z}  implies that, $\{\mathcal{A}_0^{(j)},\mathcal{A}_1^{(j)}\}=0$ for $j\in\{2,3,\cdots n\}$ as well. 

Thus far we have proved the anti-commutation relations and the fact that when the states $\ket{\psi^{(j)}_{0|x_j}}$ and $\ket{\psi^{(j)}_{1|x_j}}$ are orthogonal, the operators $\mathcal{A}^{(j)}_{x_j}$ must be constructed as specific linear combinations of Pauli matrices. This requirement arises because only such combinations can preserve the orthogonality while acting on quantum states in a way that maintains their distinct identities within the Hilbert spaces. Thus, one can always find a set of single-qubit unitaries $\{U_1,U_2,\cdots,U_n\}$ such that
\begin{equation}
      U_1 \mathcal{A}^{(1)}_{x_1} U_1^\dag =\frac{1}{\sqrt{2}}[{\sigma_X}+(-1)^{x_1}\sigma_Z]
\end{equation}
and
\begin{equation}
\begin{split}
        U_j\mathcal{A}^{(j)}_0 U_j^\dag=\sigma_X, \quad U_j\mathcal{A}^{(j)}_1 U_j^\dag=\sigma_Z,
\end{split}
\end{equation}
for all $ j \in \{2,3,\cdots,n\}$.
   
   If the $\mathcal{A}_{x_j}^{(j)}$ operators align to the given set above, the communicated qubit messages must correspond to those as given in \eqref{first set} and \eqref{second set} upto some local unitary transformation.
   
We can now say that the expectation value of the operator $\mathcal{W}_s$ must be equal to $2\sqrt{2}(n-1)$ when measured with respect to any element from the basis set $\{|\xi_{s,s'}\rangle\}_{s'}$, i.e.,
\begin{equation}
    \langle \xi_{s,s'}|\mathcal{W}_s| \xi_{s,s'} \rangle = 2\sqrt{2}(n-1) \textrm{ } \textrm{ $\forall s'$.}
\end{equation}
In other words, $\ket{\xi_{s,s'}}$ are eigenstates of $\mathcal{W}_s$ with $2\sqrt{2}(n-1)$ as their eigenvalue. However, the entire set $\{|\xi_{s,s'}\rangle\}_{s'}$ cannot have $2\sqrt{2}(n-1)$ as its corresponding eigenvalue since $\mathcal{W}_s$ has other eigenvalues as well. In fact, we will show that the eigenvalue, $2\sqrt{2}(n-1)$ is non degenerate and hence for any given $s$, only one $\ket{\xi_{s,s'}}$ corresponds to $2\sqrt{2}(n-1)$, thus making the measurement rank-1 and projective.

As we have shown that $\tr(\mathcal{W}_s\mathcal{M}_s)=2\sqrt{2}(n-1)$ implies  $U_1\mathcal{A}_{x_1}^{(1)}U_1^\dag=1/\sqrt{2}\left(\sigma_X+(-1)^{x_1}\sigma_Z\right)$, and $U_j\mathcal{A}_{0}^{(j)}U_j^\dag=\sigma_X$ and  $U_j\mathcal{A}_{1}^{(j)}U_j^\dag=\sigma_Z$ for $j \in \{2,3,\cdots,n\}$, then if one expresses $\mathcal{W}_s$ in the computational basis, one can write $\left(\bigotimes_j U_j\right)\mathcal{W}_s\left(\bigotimes_j U_j^\dag\right)=\tilde{\mathcal{W}}_s$ as
\begin{equation} \label{X-matrix}
        \tilde{\mathcal{W}}_s= \begin{bmatrix}
           \alpha_1 & 0 & \cdots & 0 & \beta_1\\
           0 & \alpha_2 & \cdots & \beta_2 & 0\\
        \vdots & \cdots  & \vdots & \cdots &\vdots\\
        \beta_{2^n} & 0 & \cdots & 0 & \alpha_{2^n}\\
        \end{bmatrix},
    \end{equation}
where the counter-diagonal elements $\{\beta_j\}_{j=1}^{2^n}$ are all equal to $\sqrt{2}(n-1)(-1)^{s_1}$, and the diagonal elements $\{\alpha_j\}_{j=1}^{2^n}$ are
\begin{equation}
    \begin{matrix}
        \alpha_1=\sqrt{2}\left((-1)^{s_2}+(-1)^{s_3}+\cdots+(-1)^{s_n}\right),\\
        \alpha_2=\sqrt{2}\left((-1)^{s_2}-(-1)^{s_3}+\cdots+(-1)^{s_n}\right),\\
        \vdots\\
        \alpha_{2^{n-1}}=-\sqrt{2}\left((-1)^{s_2}+(-1)^{s_3}+\cdots+(-1)^{s_n}\right),\\
        \alpha_{2^{n-1}+1}=\alpha_{2^{n-1}},\\
        \alpha_{2^{n-1}+2}=\alpha_{2^{n-1}-1},\\
        \vdots\\
        \alpha_{2^n}=\alpha_1.
    \end{matrix}
\end{equation}
The characteristic equation of $\tilde{\mathcal{W}}_s$, $|\tilde{\mathcal{W}}_s-\mu\mathbbm{1}|=0$ assumes the form
\begin{equation} \label{char_eq}
    \prod_{j=1}^{2^{n-1}}\left((\alpha_j-\mu)^2-\left(\sqrt{2}(n-1)(-1)^{s_1}\right)^2\right)=0.
\end{equation}
A more useful way to represent diagonal elements is to write them as $\alpha_{s'_2,s'_3,\cdots,s_n'}=\sqrt{2}\sum_{j=2}^n(-1)^{s_j'\oplus s_j}$ where $s'_j=0 \textrm{ or } 1$ $\forall j\in \{2,\cdots,n\}$. This succinctly captures the fact that there are $2^{n-1}$ unique diagonal elements, and using this representation, we can write the eigenvalues as
\begin{equation} \label{eigv}
    \mu_{s,s'}=\sqrt{2}\sum_{j=2}^n(-1)^{s_j'\oplus s_j}+\sqrt{2}(n-1)(-1)^{s_1' \oplus s_1}.
\end{equation}
The bit $s'_1$ is introduced to take care of the fact that \eqref{char_eq} is a product of difference of squares. This allows us to write $(s_1',s_2',\cdots,s_n')$ as an $n$-bit binary number $s'=s_n's_{n-1}'\cdots s_2's_1'$. It is evident that there are $2^n$ eigenvalues, with $2\sqrt{2}(n-1)$ being the maximum eigenvalue, and it is necessarily non degenerate. That is so because for \eqref{eigv} to yield $2\sqrt{2}(n-1)$, the only way that can happen is if $s'=s$, for any other choice of $s'$, the eigenvalue will be less than $2\sqrt{2}(n-1)$. \re This can be easily seen. Consider, for example, the scenario where $s'\neq s$ which implies at least one $s'_j\neq s_j$ for some $j\in\{1,2,\cdots,n\}$. All those corresponding terms in \eqref{eigv} are negative since $(-1)^{s'_j\oplus s_j}=-1$ whenever $s'_j\neq s_j$. Thus, the only way every term in \eqref{eigv} is positive is if $s'_j=s_j \textrm{ }\forall j$ which implies that only for $s'=s$ one gets the maximum eigenvalue. \blk \\
\setlength{\parskip}{3pt}

The non degeneracy of the maximum eigenvalue implies that the measurement is rank-one projective and unique. Let that measurement be denoted by $\{\ket{\tilde{\xi}_s}\}_s$. Furthermore, let's assume that $\ket{\xi_s}=\left(\bigotimes_jU_j^{\dag}\right)\ket{\tilde{\xi}_s}\left(\bigotimes_jU_j\right)$ is of the form $1/\sqrt{2}\left(\ket{x_1'x_2'\cdots x_n'}\pm \ket{\overline{x}_1'\overline{x}_2'\cdots \overline{x}_n'}\right)$ where $x_j'=0$ or $1$, and $\overline{x}_j'$ is its compliment $\forall j \in \{1,\cdots,n\}$. $\bra{\xi_s}\mathcal{W}_s\ket{\xi_s}$ can then be divided into four terms, two of which we shall call diagonal terms, $(1/2)\bra{x_1'\cdots x_n'}\mathcal{W}_s\ket{x_1'\cdots x_n'}$ and $(1/2)\bra{\overline{x}_1'\cdots \overline{x}_n'}\mathcal{W}_s\ket{\overline{x}_1'\cdots \overline{x}_n'}$ and two of which we shall call the off-diagonal terms, $(1/2)\bra{x_1'\cdots x_n'}\mathcal{W}_s\ket{\overline{x}_1'\cdots \overline{x}_n'}$ and $(1/2)\bra{\overline{x}_1'\cdots \overline{x}_n'}\mathcal{W}_s\ket{x_1'\cdots x_n'}$. A bit of careful observation reveals that the non zero off-diagonal terms come from the $(n-1)\ (-1)^{s_1}\left(\mathcal{A}_0^{(1)}+\mathcal{A}_1^{(1)}\right)\otimes\bigotimes_{j=2}^n\mathcal{A}_0^{(j)}$ part of $\mathcal{W}_s$ regardless of what the $\{x'_j\}_{j=1}^{n}$ are. In contrast, the non zero diagonal terms come from the $\sum_{j=2}^n (-1)^{s_j}\left(\mathcal{A}_0^{(1)}-\mathcal{A}_1^{(1)}\right)\otimes\mathcal{A}_1^{(j)}$ part of $\mathcal{W}_s$ for the case where $x'_j=s_j$ for all $j \in \{2,\cdots,n\}$. Finally, the contributions from the diagonal and off-diagonal terms are added in such a way that it results in $2\sqrt{2}(n-1)$, which depends on $s_{1}$. Combining all these factors, we arrive at
\begin{equation}
    \ket{\xi_s}=\frac{1}{\sqrt{2}}\Big(\ket{0 \textrm{ } s_2 \textrm{ }\cdots  \textrm{ }s_n}+(-1)^{s_1}\ket{1 \textrm{ }\overline{s_2} \textrm{ }\cdots\textrm{ }\overline{s_n}}\Big).
\end{equation}
This completes the ideal self-testing of the messages and the entangled measurement.
\end{proof}
\subsection{Robust Self-testing of the measurements} \label{robustness}
In the preceding subsection, we had considered an ideal self-testing scenario where the maximum quantum value of $\mathcal{S}$ guarantees that the measurements employed must be maximally entangled GHZ basis. However, in practice, the measurements may deviate from the ideal self-tested ones, suggesting that the observed value of $\mathcal{S}$ falls short of the maximum quantum value. This reflects imperfections in measurement implementation and motivates the need to analyze the robustness of the certification of the entangled basis, which we would quantify using fidelity.

 We use a similar method prescribed in \cite{Baccari_2020,Kan} to derive the fidelity limits. Let us recall the basic terminologies discussed in \cite{Kan}. Considering the definition of extractibility, the average fidelity $\mathcal{F}$ for an arbitrary set of measurements $\{\mathcal{M}_s\}_s$, with the ideal measurements $\{\ket{\xi_s}\!\bra{\xi_s}\}_s$, is defined as 
\begin{equation}\label{fidelity}
\mathcal{F}{(\{\mathcal{M}_{s}\}_s)}=\max_{\{\Lambda\}} \frac{1}{2^n} \sum_{s}F(\ket{\xi_s}\!\bra{\xi_s},\Lambda[\mathcal{M}_{s}]),
\end{equation}  where $\Lambda=\bigotimes_{j=1}^n\Lambda^{(j)}$ is a valid local quantum channel. The maximization is taken over all local quantum channels $\{\Lambda\}$. We notice that the fidelity $F$ in \eqref{fidelity} is expressed as
\begin{equation} \label{fidelity of f}  F(\ket{\xi_s}\!\bra{\xi_s},\Lambda[\mathcal{M}_{s}])=\tr{(\mathcal{M}_s\Lambda^\dag[\ket{\xi_s}\!\bra{\xi_s}])},
\end{equation}
where $\Lambda^\dag$ is the dual map of the local quantum channel $\Lambda$. Our aim is to derive a lower bound of $\mathcal{F}$ involving a minimization on the set of measurements, given the value of $\mathcal{S}$.

However, before proceeding further, we take a brief detour. Consider the quantum states prepared by two senders, $\rho^{(1)}_{a_1|x_1}$ and $\rho^{(2)}_{a_2|x_2}$, corresponding to the preparation procedures of the senders, $A^{(1)}$ and $A^{(2)}$, respectively. 

    Taking into account the four POVM measurement operators $\{\mathcal{M}_s\}_s$, \eqref{CHSH operator} can be reformulated as
    \begin{equation} \label{decompo}
    \begin{split}
    \mathcal{S}=\frac{1}{8\sqrt{2}}\tr \Big[\Big(\mathcal{Q}_{00}+\mathcal{Q}_{10}\Big)f_1(\mathcal{M})
    +\Big(\mathcal{Q}_{01}-\mathcal{Q}_{11}\Big)f_2(\mathcal{M})\Big],
    \end{split}
    \end{equation}
    where
    \begin{equation}\label{qx1x2}
        \begin{split}
            \mathcal{Q}_{x_1x_2}=\sum_{a_1,a_2=0,1}(-1)^{a_1 \oplus a_2}\rho^{(1)}_{a_1|x_1}\otimes\rho^{(2)}_{a_2|x_2}, \quad x_1,x_2 \in \{0,1\},
        \end{split}
    \end{equation}
    and
    \begin{equation}
        f_j(\mathcal{M})=\sum_s(-1)^{s_j}\mathcal{M}_s, \quad j\in \{1,2\}.
    \end{equation}
    If one were to take any $\mathcal{Q}_{x_1x_2}$ and write it out explicitly, one would have a linear combination of $4$ terms, $\mathcal{Q}_{x_1x_2}=\rho_{0|x_1}^{(1)}\otimes\rho_{0|x_2}^{(2)}-\rho_{0|x_1}^{(1)}\otimes\rho_{1|x_2}^{(2)}-\rho_{1|x_1}^{(1)}\otimes\rho_{0|x_2}^{(2)}+\rho_{1|x_1}^{(1)}\otimes\rho_{1|x_2}^{(2)}$, such that $\tr(\mathcal{Q}_{x_1x_2}f_j(\mathcal{M}))$ can be written as,
    \begin{equation} \label{tr_Q_f_M}
        \begin{split}
\tr(\mathcal{Q}_{x_1x_2}f_j(\mathcal{M}))&=\tr\Big[\rho_{0|x_1}^{(1)}\otimes\left(\rho_{0|x_2}^{(2)}-\rho_{1|x_2}^{(2)}\right)f_j(\mathcal{M})\Big]\\
&+\tr\Big[\left(\rho_{1|x_1}^{(1)}-\rho_{0|x_1}^{(1)}\right)\otimes\rho_{1|x_2}^{(2)}f_j(\mathcal{M})\Big].
\end{split}
    \end{equation}

    Since $f_j(\mathcal{M})$ are linear combinations of POVM elements, they are Hermitian operators, and thus \eqref{tr_Q_f_M} is maximized when each such $\rho_{a_1|x_1}^{(1)}\otimes\rho_{a_2|x_2}^{(2)}$ are pure and eigenstates of $f_j(\mathcal{M})$. Not only that, $\rho_{0|x_1}^{(1)}\otimes\rho_{0|x_2}^{(2)}$ and $\rho_{1|x_1}^{(1)}\otimes\rho_{1|x_2}^{(2)}$ must correspond to the maximum eigenvalue of $f_j(\mathcal{M})$ while $\rho_{0|x_1}^{(1)}\otimes\rho_{1|x_2}^{(2)}$ and $\rho_{1|x_1}^{(1)}\otimes\rho_{0|x_2}^{(2)}$ must correspond to the minimum eigenvalue. This automatically implies that

    \begin{equation} \label{ortho_cond}
        \tr\left(\rho_{0|x_j}^{(j)}\rho_{1|x_j}^{(j)}\right)=0, \quad j\in\{1,2\}.
    \end{equation}
    
Since $\mathcal{S}$ is a sum of such terms as shown in \eqref{tr_Q_f_M}, for a given POVM set $\{\mathcal{M}_s\}_s$, $\mathcal{S}$ can be maximized only when each sender $A^{(j)}$ for $j\in \{1,2\}$ transmits messages which obey \eqref{ortho_cond}, that is they must correspond to anti-podal Bloch vectors. One can arrive at a similar conclusion for the $n-$senders, one receiver scenario, which has been discussed in Appendix \ref{structure_of_metric}. This implies that while formulating a minimum bound on the fidelity of the measurement, it is sufficient to consider the messages as antipodal in nature. A precise explanation of this has been provided in the proof of the theorem below.

\begin{thm}\label{robust1}
    Given a non-optimal quantum value of $\mathcal{S}\geq 1-\epsilon$ as defined in \eqref{SGen}, one can define a local map $\Lambda=\bigotimes_{j=1}^n\Lambda^{(j)}$ such that one can lower bound the average fidelity of $\{\mathcal{M}_s\}_s$, $\mathcal{F}(\{\mathcal{M}_s\}_s)$, \re given that $n\leq7$, \blk in the following manner,
    \begin{equation} \label{n_input_fid}
        \mathcal{F}\left(\{\mathcal{M}_s\}_s\right)\geq \left(r(n-1)2\sqrt{2}+\mu\right)-r(n-1)2\sqrt{2}\epsilon,
    \end{equation}
    where $r$ and $\mu \in \mathbb{R}$ such that $r(n-1)2\sqrt{2}+\mu=1$. For $n=2$, $r=(4+5\sqrt{2})/16$ and $\mu=-(1+2\sqrt{2})/4$ resulting in
    \begin{equation} \label{2_input_fid}
        \mathcal{F}\left(\{\mathcal{M}_0,\mathcal{M}_1\}\right)\geq 1-\frac{4+5\sqrt{2}}{4\sqrt{2}} \epsilon.
    \end{equation}
\end{thm}
\vspace{-3pt}
\begin{proof}
Before we begin with the main proof, we emphasize to the readers that to determine the fidelity bound of the measurement, it is sufficient to consider the messages being sent by the senders as antipodal, that is, in accordance with \eqref{ortho_cond}. 

We show this by contradiction. Assume that the minimum fidelity bound that is tight, is due to non-antipodal messages. To be more precise, suppose that for a measurement $\{\mathcal{M}_s\}_s$, and $\{\mathcal{W}_s\}_s$ operators constructed out of non-antipodal messages, $\mathcal{S}=1-\epsilon$ and the corresponding fidelity bound $\mathcal{F}(\{\mathcal{M}_s\}_s)\geq1-f(\epsilon)$ where $\lim_{\epsilon\rightarrow0}f(\epsilon)=0$. From the discussions following the equations \eqref{decompo}-\eqref{tr_Q_f_M}, we have seen that for a given measurement, the best $\mathcal{S}$ is achieved by antipodal messages. Let $\{\overline{\mathcal{W}}_s\}_s$ be the operators constructed out of antipodal messages for this set of measurements $\{\mathcal{M}_s\}_s$. Then we know that

\begin{equation}
\begin{split}
&\sum_s\tr(\mathcal{M}_s\mathcal{W}_s)\leq \sum_s\tr(\mathcal{M}_s\overline{\mathcal{W}}_s) \quad \textrm{which implies that,}\\
&\sum_s\tr(\mathcal{M}_s\mathcal{W}_s)=\delta\sum_s\tr(\mathcal{M}_s\overline{\mathcal{W}}_s) \quad \textrm{for some $\delta \in [0,1]$}.
\end{split}
\end{equation}

To make $\overline{\mathcal{S}}=1-\epsilon$ using $\{\overline{\mathcal{W}}_s\}_s$, we introduce some error in the corresponding measurement such that if we consider $\overline{\mathcal{M}}_s=\delta\mathcal{M}_s+(1-\delta)\mathbbm{1}/2^n$ then
\re
\begin{equation}
    \begin{split}
        \overline{\mathcal{S}}&=\frac{1}{(2\sqrt{2}(n-1))}\sum_s\tr(\overline{\mathcal{M}}_s\overline{\mathcal{W}}_s)\\
        &=\frac{\delta}{(2\sqrt{2}(n-1))}\sum_s\tr(\mathcal{M}_s\overline{\mathcal{W}}_s)\\
        &=1-\epsilon.
    \end{split}
\end{equation}

 In the above equation, one uses the fact that $\tr(\overline{\mathcal{M}}_s\overline{\mathcal{W}}_s)=\delta\tr(\mathcal{M}_s\overline{\mathcal{W}}_s)+((1-\delta)/2^n)\tr(\overline{\mathcal{W}}_s)$ and the fact that $\tr(\overline{\mathcal{W}}_s)=0$. This is so because using \eqref{w_i_def}, one can see that $\overline{\mathcal{W}}_s$ is a sum of terms which are differences of density matrices of pure states, and hence trace of every such term cancels out making $\overline{\mathcal{W}}_s$ traceless. \blk The fidelity bound of the modified measurements become
\begin{equation}
\begin{split}
    \frac{1}{2^n}\sum_s\bra{\xi_s}\Lambda(\overline{\mathcal{M}}_s)\ket{\xi_s}&=\delta\frac{1}{2^n}\sum_s\bra{\xi_s}\Lambda(\mathcal{M}_s)\ket{\xi_s}+\frac{1-\delta}{2^n}\\
    &\geq \delta(1-f(\epsilon))+\frac{1-\delta}{2^n}.
    \end{split}
\end{equation}
The bound achieved above is less than $1-f(\epsilon)$ \re and can be checked by simply subtracting $\delta(1-f(\epsilon))+(1-\delta)/2^n$ from $1-f(\epsilon)$ which results in $(1-\delta)(1-1/2^n-f(\epsilon))$. Now, $f(\epsilon)<1-1/2^n$ for $\epsilon< \textrm{ some }  \epsilon_0$ since $\lim_{\epsilon\rightarrow0}f(\epsilon)=0$.  \blk Thus, we arrive at a contradiction. One cannot arrive at a minimum tight fidelity bound from non-antipodal messages, or in other words, it is sufficient to consider the messages sent by the senders as antipodal to derive a meaningful fidelity bound of measurement.

Now having established that, we shall derive \eqref{n_input_fid} by using operator inequalities like the one introduced in \cite{Kan}
\begin{equation}\label{operator inequality}
    \mathcal{K}_s\geq r\mathcal{W}_s+\mu\mathbbm{1},
\end{equation}
 where $r,\mu\in \mathbb{R}$ and $ \mathcal{K}_s=\Lambda^\dag[\ket{\xi_s}\!\bra{\xi_s}].$ Tracing out both sides of the inequality after multiplying it with the corresponding realized measurement operator $\mathcal{M}_s$, \eqref{operator inequality} takes the form,
 \begin{equation}
     F(\ket{\xi_s}\!\bra{\xi_s},\Lambda[\mathcal{M}_{s}])\geq r\tr(\mathcal{M}_s\mathcal{W}_s) + \mu\tr(\mathcal{M}_s).
 \end{equation}
 Since the messages are pure and antipodal (i.e, $\rho_{0|x_j}^{(j)}$ is orthogonal to  $\rho_{1|x_j}^{(j)}$), we can always parametrize the operators $\mathcal{A}^{(j)}_{x_j}$ as
\begin{equation}
\begin{split}
    &\mathcal{A}_{x_1}^{(1)}=\cos\alpha_1\sigma_X + (-1)^{x_1}\sin\alpha_1\sigma_Z,\\&
    \mathcal{A}_{x_j}^{(j)}=\cos\alpha_j\sigma_A + (-1)^{x_j}\sin\alpha_j\sigma_B,  \quad \forall j=2,\cdots,n
\end{split}    
\end{equation}
where ${\sigma_A} =(\sigma_X + \sigma_Z)/\sqrt{2} $ and $\sigma_B=(\sigma_X - \sigma_Z)/\sqrt{2} $ and $\alpha_j\in[0,\pi/2] \quad \forall j \in \{1,2,\cdots,n\}$. The operators $\{\mathcal{W}_s\}_s$ now depend on the angles $\{\alpha_j\}_j$. Equation \eqref{operator inequality} suggests that every $\mathcal{K}_s$ must also be functions of $\{\alpha_j\}_{j}$ such that $\Lambda=\Lambda(\alpha_1,\cdots,\alpha_n)$. One can construct such a channel if every single-qubit channel $\Lambda^{(j)}$ becomes parametrized by $\alpha_j$; they are defined as channels similar to the one used \re in Eq.(9) \blk of \cite{Kan}
\begin{equation}\label{unital channel}
    \Lambda^{(j)}(x)[\rho]=\frac{1+g(x)}{2}\rho+\frac{1-g(x)}{2}\Gamma^{(j)}(x)\rho\Gamma^{(j)}(x),
\end{equation}
where $\rho$ denotes a single-qubit quantum state. The function $g(x)=(1+\sqrt{2})(\sin x+\cos x-1)$ denotes the dependence of the channel on the angles $\alpha_j$. The operators $\{\Gamma^{(j)}\}_j$ are defined as
\begin{equation} \label{gamma_def}
\begin{split}
    &\Gamma^{(1)}(x)=\sigma_X\quad x\leq\pi/4 \qquad \Gamma^{(1)}(x)=\sigma_Z\quad x>
\pi/4,\\& 
     \Gamma^{(j)}(x)=\sigma_A\quad x\leq\pi/4  \qquad \Gamma^{(j)}(x)=\sigma_B\quad x>
\pi/4,
\end{split}
\end{equation}
for $j=2,\cdots,n$. From \eqref{unital channel} it is evident that $\Lambda$ is self-dual, and hence from now on we shall drop the dagger ($\dag$). Notice that 
it has already been proved in \cite{Baccari_2020} that for all possible values of $\alpha_j$, the inequality \eqref{operator inequality} is valid for some choice of $r,\mu\in \mathbb{R}$ for $s=00\cdots0$ given that $n\leq7$. We shall show that given the fidelity bound of one such $\mathcal{M}_s$, it is possible to infer the fidelity bounds of the rest as well.

Any two maximally entangled GHZ measurements with the same number of outcomes and acting on a finite-dimensional Hilbert space can be connected by a local unitary transformation. This unitary equivalence also extends to the corresponding operator equations, implying a structural correspondence between the measurements. Consider the local unitary, $U_{s\rightarrow s'}=\bigotimes_{j=1}^nU^{(j)}_{s\rightarrow s'}$ which transforms $\ket{\xi_s}$ to $\ket{\xi_{s'}}$. They are defined as,
\begin{equation}
    \begin{split}
        &U_{s\rightarrow s'}^{(1)}=\mathbbm{1} \textrm{ if } s_1=s_1', \qquad U_{s\rightarrow s'}^{(1)}=\sigma_Z \textrm{ if } s_1\neq s_1', \textrm{ and}\\
        &U_{s\rightarrow s'}^{(j)}=\mathbbm{1} \textrm{ if } s_j=s_j', \qquad U_{s\rightarrow s'}^{(j)}=\sigma_X \textrm{ if } s_j\neq s_j',
    \end{split}
\end{equation}
for all $j\in \{2,\cdots,n\}$. The same unitary transforms $\mathcal{W}_s$ to $\mathcal{W}_{s'}$ such that
\begin{equation}
    \begin{split}
        &\ket{\xi_{s'}}=U_{s\rightarrow s'}\ket{\xi_s} \quad \textrm{and,}\\
        &\mathcal{W}_{s'}=U_{s\rightarrow s'}\mathcal{W}_sU^{\dag}_{s\rightarrow s'}.
    \end{split}
\end{equation}
Consider $\mathcal{K}_s-r\mathcal{W}_s$, noting that $\mathcal{K}_s=\Lambda[\ket{\xi_s}\!\bra{\xi_s}]$ which has been shown to be lower-bounded by $\mu\mathbbm{1}$ for $s=00\cdots0$. Now let's consider $\mathcal{K}_{s'}-r\mathcal{W}_{s'}$ where $\mathcal{K}_{s'}=\Lambda[\ket{\xi_{s'}}\!\bra{\xi_{s'}}]$. We can write
\begin{equation} \label{Ks-rWs}
    \begin{split}
        &\Lambda[\ket{\xi_{s'}}\!\bra{\xi_{s'}}]-r\mathcal{W}_{s'}=\\
        &\Lambda[U_{s\rightarrow s'}\ket{\xi_{s}}\!\bra{\xi_{s}}U^{\dag}_{s\rightarrow s'}]-rU_{s\rightarrow s'}\mathcal{W}_{s}U_{s\rightarrow s'}^{\dag}.
    \end{split}
\end{equation}
It can be shown that $\Lambda[U_{s\rightarrow s'}\ket{\xi_{s}}\!\bra{\xi_{s}}U^{\dag}_{s\rightarrow s'}]=U_{s\rightarrow s'}\Lambda[\ket{\xi_{s}}\!\bra{\xi_{s}}]U_{s\rightarrow s'}^{\dag}$, then \eqref{Ks-rWs} becomes
\begin{equation}
    \mathcal{K}_{s'}-r\mathcal{W}_{s'}=U_{s\rightarrow s'}\left(\mathcal{K}_s-r\mathcal{W}_s\right)U_{s\rightarrow s'}^{\dag}\geq \mu \mathbbm{1}.
\end{equation}
The above equation is easily achieved because $\Gamma^{(j)}=\pm U_{s\rightarrow s'}^{(j)}\Gamma^{(j)}\left(U_{s\rightarrow s'}^{(j)}\right)^{\dag}$ $\forall j\in\{1,2,\cdots,n\}$. Thus, one can generate operator inequalities of the form \eqref{operator inequality} for all $s$ from the operator inequality for $s=00\cdots0$. The average fidelity, as defined in \eqref{fidelity}, then results in \eqref{n_input_fid}.
 For $n=2$, it was derived analytically in \cite{Kan} that $r=(4+5\sqrt{2})/16\approx 0.6919$, and $\mu=-(1+2\sqrt{2})/4 \approx -0.9571$, thus establishing \eqref{2_input_fid}.
\end{proof}

The fidelity bound must be greater than $1/2$ for the verifier to infer anything meaningful about the degree of entanglement present in the realized measurement. This puts an upper bound on the error $\epsilon$ such that $\epsilon\leq1/\left(4\sqrt{2}r(n-1)\right)$ and specifically for the two-sender, one-receiver scenario, \re $\epsilon\leq (2\sqrt{2}/(4+5\sqrt{2}))\approx0.25$. For a graphical representation of the fidelity bounds with the deviation $\epsilon$ please refer to FIG. \ref{partial bell basis figuire}.\blk

\section{Robust Self-testing of partial Bell basis measurements}\label{Robust Self-testing of partial Bell basis measurements}

It is well established that, using linear optics, a complete Bell basis measurement cannot be implemented on the polarization degrees of freedom of two distinct photons without the use of ancillary systems \cite{NogoBSM1,NogoBSM2}. However, a partial Bell basis measurement, for example, the three-outcome measurement defined by $\{\ket{\phi^+}\!\bra{\phi^+},\ket{\phi^-}\!\bra{\phi^-}, \ket{\psi^+}\!\bra{\psi^+}+\ket{\psi^-}\!\bra{\psi^-}\} $, where, $\ket{\phi^{\pm}}=\frac{1}{\sqrt{2}}[\ket{00}\pm\ket{11}]$, and  $\ket{\psi^{\pm}}=\frac{1}{\sqrt{2}}[\ket{01}\pm\ket{10}]$, can be realized readily. \re The partial Bell-basis measurement has been effectively realized in several important contexts, including quantum teleportation, dense coding, quantum repeaters, the Hong–Ou–Mandel effect, and many others \cite{HOM,Teleportation_nature,Teleportation_prl,Densecoding_prl,Repeater_prl,Prabhakar_pra,expt-prl}. \blk  Motivated by this practical constraint, we adapt our self-testing protocol to robustly certify this partial Bell basis measurement, making it suitable for implementation in linear optical systems. 

Let us consider a specific communication scenario involving two senders, $A^{(1)}$and $A^{(2)}$ and a single receiver, $R$. As before, each sender $A^{(j)}$, for $j\in\{1,2\}$, receives two inputs: $x_j\in\{0,1\}$ and $a_j\in\{0,1\}$. Based on the inputs they receive, the senders prepare and send messages $m_d(x_j,a_j)$ to the receiver. However, the receiver in this setup accepts an input $k$, where $k\in\{1,2,3\}$. For $k=1$ and $k=2$, the receiver engages in a prepare-and-measure task with sender $A^{(1)}$ and relabels the inputs of $A^{(1)}$ as $x'_1=x_1\oplus a_1$ and $a'_1=a_1$ and aims to guess the $a_1'$ or $x_1'$ bit, respectively, by implementing binary measurements in each case.  For $k=3$, we continue with a similar setup as previously considered in Sec. \ref{comm task} except that the receiver has three distinct outcomes, $\{1,2,3\}$ as shown in FIG. \ref{fig:enter-label1}.

In such a scenario, we define two different success metrics. For $k=1$ and $2$ we have $\mathcal{S}^{RAC}$,
\begin{equation}
    \mathcal{S}^{RAC}=\sum_{a'_1,x'_1}\frac{1}{4}\bigg[p(a'_1|a'_1,x'_1,1)+p(x'_1|a'_1,x'_1,2)\bigg],
\end{equation}
where, $p(z|a_1',x_1',k)$ represents the statistics  involving the inputs of $A^{(1)}$ and $R$ and the outcomes $\{z\}$ that $R$ sees. For $k=3$ we have
\begin{equation}\label{partial bell basis}
     \mathcal{S}^{Comm}=\frac{1}{8\sqrt{2}}\left(W'_1+W'_2+W'_3\right),
 \end{equation}
where
 \begin{widetext}
 \vspace{-0.6cm}
     \begin{equation}
     \begin{split}
    W'_1=\sum_{a_1,a_2=0,1}&(-1)^{a_1\oplus a_2}\Big[p(1|a_1,a_2;x_1=x_2=0)+p(1|a_1,a_2;x_1=1,x_2=0)\\
    &+p(1|a_1,a_2;x_1=0,x_2=1)-p(1|a_1,a_2;x_1=x_2=1)\Big].
      \end{split}
     \end{equation}
     \begin{equation}
     \begin{split}
    W'_2=\sum_{a_1,a_2=0,1}&(-1)^{a_1\oplus a_2}\Big[p(2|a_1,a_2;x_1=x_2=0)+p(2|a_1,a_2;x_1=1,x_2=0)\\
    &-p(2|a_1,a_2;x_1=0,x_2=1)+p(2|a_1,a_2;x_1=x_2=1)\Big].
      \end{split}
     \end{equation}
     \begin{equation}
     \begin{split}
    W'_3=\sum_{a_1,a_2=0,1}(-1)^{a_1\oplus a_2\oplus 1}\Big[p(3|a_1,a_2;x_1=x_2=0)+p(3|a_1,a_2;x_1=1,x_2=0)\Big].
      \end{split}
     \end{equation}
 \end{widetext}
The maximum value of the success metric $\mathcal{S}^{RAC}$ achievable within quantum theory is equal to $(1+1/\sqrt{2})/2\approx0.85$ and for $\mathcal{S}^{Comm}$, its equal to 1.

\begin{figure}[h] 
    \centering
    \includegraphics[width=0.45\textwidth]{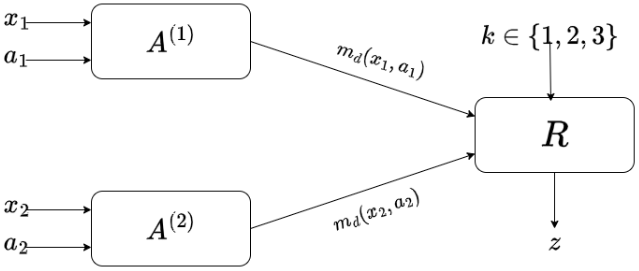}
    \caption{A schematic diagram of a two-sender one-receiver communication scenario in which $A^{(1)}$ and $A^{(2)}$, each transmits a message $m_d(x_j,a_j)$, with $d=2$ in our case, determined by the respective inputs $(x_j,a_j)$ where $j\in\{1,2\}$. These messages are sent to a receiver, who accepts an input $k\in\{1,2,3\}$ and tries to guess the first or the second bit of $A^{(1)}$ for $k=1,2$ respectively, and produces three outputs for $k=3$.}
    \label{fig:enter-label1}
\end{figure}

In a quantum communication scenario the senders $A^{(1)}$ and $A^{(2)}$ encode their input as qubit states, $\rho^{(j)}_{a_j|x_j}, j\in\{1,2\}$ respectively and transmit them to the receiver $R$, who re-labels the states from the first sender $A^{(1)}$ as $\rho^{(j)}_{a'_1|x'_1}$ for input $k=1,2$ and performs a $2\rightarrow 1$ random access code using binary-outcome measurements $\mathcal{M}^{(k)}=\mathcal{M}^{(k)}_1-\mathcal{M}^{(k)}_2$. For $k=3$, the receiver performs a measurement with a valid POVM, $\{\mathcal{M}^{(3)}_i\}_{i=1}^3$ to calculate the value of the success metric $\mathcal{S}^{Comm}$, where $W'_i=\tr({\mathcal{M}^{(3)}_i\mathcal{W'}_i})$ and $\{\mathcal{W'}_i\}_i$ are defined as
 \begin{equation} \label{partial bell operator eq}
        \begin{split}
            \mathcal{W'}_{1} &=(\mathcal{A}_{0}^{(1)}+\mathcal{A}_{1}^{(1)})\otimes \mathcal{A}_{0}^{(2)}+(\mathcal{A}_{0}^{(1)}-\mathcal{A}_{1}^{(1)})\otimes \mathcal{A}_{1}^{(2)},\\
            \mathcal{W'}_{2}&=-(\mathcal{A}_{0}^{(1)}+\mathcal{A}_{1}^{(1)})\otimes \mathcal{A}_{0}^{(2)}+(\mathcal{A}_{0}^{(1)}-\mathcal{A}_{1}^{(1)})\otimes \mathcal{A}_{1}^{(2)},\\
            \mathcal{W'}_{3}&=-2(\mathcal{A}_{0}^{(1)}-\mathcal{A}_{1}^{(1)})\otimes \mathcal{A}_{1}^{(2)},
        \end{split}
        \end{equation}
        where the operators $\mathcal{A}^{(j)}_{x_j}=\rho^{(j)}_{0|x_j}-\rho^{(j)}_{1|x_j}$.
 Now, we continue to prove that the value of $\mathcal{S}^{Comm}$ attains its maximum value when measurements are performed in partial Bell basis measurements. 
\begin{thm}\label{th4}
The optimum quantum value of $\mathcal{S}^{RAC}$ and $\mathcal{S}^{Comm}$ in \eqref{partial bell basis} are $(1+1/\sqrt{2})/2$ and $1$ respectively. If these values are achieved from an unknown set of qubit states $\{\rho^{(j)}_{a_j|x_j}\}_{a_j,x_j}$ and an unknown measurement $\{\mathcal{M}^{(3)}_i\}_i$ for $k=3$ on the tqubit system, then there exists a set of qubit unitaries $\{U_j\}_j$ such that the  states are equivalent to the ideal reference implementation as defined in Eqs. \eqref{first set} and \eqref{second set} in Theorem \ref{th1} and the measurement operators $\{\mathcal{M}^{(3)}_i\}_i$ are
    \begin{equation}
         \begin{split}
             &(U_1\otimes U_2)\mathcal{M}^{(3)}_1(U_1^\dag\otimes U_2^\dag)= \ket{\phi^+}\!\bra{\phi^+}, \\&(U_1\otimes U_2)\mathcal{M}^{(3)}_2(U_1^\dag\otimes U_2^\dag)=\ket{\phi^-}\!\bra{\phi^-},\\&(U_1\otimes U_2)\mathcal{M}^{(3)}_3(U_1^\dag\otimes U_2^\dag)=\ket{\psi^+}\!\bra{\psi^+}+\ket{\psi^-}\!\bra{\psi^-}.
         \end{split}    
    \end{equation}
The measurement for $k=3$ on the receiver's end can be shown to be robust to noise in case of non-optimal values of $\mathcal{S}^{Comm}$ when lower bounded by $1-\epsilon$ (i.e. $\mathcal{S}^{Comm}\geq1-\epsilon)$ and $\mathcal{S}^{RAC}$ using a local map $\Lambda=\Lambda^{(1)}\otimes\Lambda^{(2)}$, such that both fidelities $F(|\phi^+\rangle\!\langle \phi^+|,\Lambda[\mathcal{M}^{(3)}_1]), F(|\phi^-\rangle\!\langle \phi^-|,\Lambda[\mathcal{M}^{(3)}_2]) $, defined in \eqref{fidelity of f}, are lower bounded by 
\begin{equation}\label{fidelity partial}
\begin{split}
     &  F(|\phi^+ \rangle\!\langle \phi^+| ,\Lambda[\mathcal{M}^{(3)}_1]) \ \left(\text{or } F(|\phi^-\rangle\!\langle \phi^-|,\Lambda[\mathcal{M}^{(3)}_2])\right) \\
     \geq & \ 1-\frac{8\sqrt{2}\epsilon}{3}\left(r-\frac{\mu}{\sqrt{2}}\right)
        \\&+\mu \cos^{-1}\left(2\sqrt{2}\left(\mathcal{S}^{RAC}-
        \frac{1}{2}\right)\right)\left(2-\frac{4\epsilon}{3}\right),
\end{split}
\end{equation}
where $r={(4+5\sqrt{2})}/{16}$ and $\mu=-{(1+2\sqrt{2})}/{4}$. 
\end{thm}
\begin{proof}
\textit{For ideal self-testing:} Let us write the qubit preparations for both senders as $\rho^{(j)}_{a_j|x_j}=(\mathbbm{1}+\vec{m}^{(j)}_{a_j|x_j}\cdot\vec{\sigma})/2$  where $\vec{m}^{(j)}_{a_j|x_j}$ denotes the Bloch vector ($|\vec{m}^{(j)}_{a_j|x_j}|=1$, if pure and $|\vec{m}^{(j)}_{a_j|x_j}|<1$, if mixed) and $\vec{\sigma}=(\sigma_X,\sigma_Y,\sigma_Z)$ denotes the vector of Pauli matrices. For $k=1,2$, let us choose qubit states sent by $A^{(1)}$ to the receiver after relabeling as $\rho^{(1)}_{a'_1|x'_1}=(\mathbb{1}+\vec{m}^{(1)}_{a'_1|x'_1}\cdot\vec{\sigma})/2$. Then, the value of $\mathcal{S}^{RAC}$ can be upper bounded as
\begin{equation}\label{srac}
    \mathcal{S}^{RAC}\leq \frac{1}{2}+\frac{1}{8\sqrt{2}}[\sqrt{\gamma+\beta}+\sqrt{\gamma-\beta}],
\end{equation}
where
\begin{equation}
\begin{split}
    &\gamma=\frac{1}{2}\sum_{a'_j,x'_j}|\vec{m}^{(1)}_{a'_j|x'_j}|^2-\vec{m}^{(1)}_{0|0}\cdot\vec{m}^{(1)}_{1|1}-\vec{m}^{(1)}_{0|1}\cdot\vec{m}^{(1)}_{1|0} \quad \textrm{and,}
    \\
    &\beta=(\vec{m}^{(1)}_{0|0}-\vec{m}^{(1)}_{1|1})\cdot(\vec{m}^{(1)}_{0|1}-\vec{m}^{(1)}_{1|0}).
\end{split} 
\end{equation}
When the receiver observes the maximum value of $\mathcal{S}^{RAC}$, the set of four prepared states  on the $A^{(1)}$'s end must be equivalent to the four ideal states given in Theorem \ref{th1} and the receiver carries out binary measurements $\mathcal{M}^{(k)}$, where $\mathcal{M}^{(1)}=\sigma_X$ and $\mathcal{M}^{(2)}=\sigma_Z$ \cite{PhysRevA.98.062307}.

The two operators $\mathcal{W'}_1$ and $\mathcal{W'}_2$ are structurally analogous to the two operators $\mathcal{W}_1$ and $\mathcal{W}_2$ mentioned in the communication scenario involving two senders and a receiver in \eqref{chsh OPERATORS}.  By the same argument shown in Theorem \ref{th1} we can say that $\|\mathcal{W'}_1\|$ and $\|\mathcal{W'}_2\|$ are bounded by $2\sqrt{2}$. Now for $\mathcal{W'}_3$, since ideal self-testing of the messages has been established on the $A^{(1)}$'s side, and noting that the observable $\mathcal{A}^{(2)}_{1}$ can be taken as $\vec{a}\cdot\vec{\sigma}$ where $\vec{a}=(\vec{m}^{(2)}_{0|1}-\vec{m}^{(2)}_{1|1})/2$ and $|\vec{a}|\leq1$, we can say that $\|\mathcal{W'}_3\|$ can be bounded by $2\sqrt{2}$. 

Now consider the POVM $\{\mathcal{M}^{(3)}_i\}_i$. When self-tested, the expectation values of the operator equations \eqref{partial bell operator eq} when measured with respect to such a measurement operator $\{\mathcal{M}^{(3)}_i\}_i$ must be equal to $2\sqrt{2}$ for $i=1,2$ and $4\sqrt{2}$ for $i=3$, i.e, 
\begin{equation}\label{trace equations partial}
\begin{split}
   &\tr({\mathcal{M}^{(3)}_1\mathcal{W'}_1})=2\sqrt{2},\\
   & \tr({\mathcal{M}^{(3)}_2\mathcal{W'}_2})=2\sqrt{2},\\
   & \tr({\mathcal{M}^{(3)}_3\mathcal{W'}_3})=4\sqrt{2}.
\end{split}
\end{equation}
Following the same argument as given in Theorem \ref{th1}, we can arrive at the same conclusion \eqref{tr_mi_geq1} for $\mathcal{M}^{(3)}_1$ and $\mathcal{M}^{(3)}_2$ i.e. $\tr{\mathcal{M}^{(3)}_i}\geq1$ for $i=1,2$. Using the fact that $\|\mathcal{W'}_3\|\leq2\sqrt{2}$ and \eqref{trace equations partial} we can say that $\tr{\mathcal{M}^{(3)}_3}\geq2$. Since $\{\mathcal{M}^{(3)}_i\}_i$ forms a POVM and $\tr{(\sum_{i=1}^{3}\mathcal{M}^{(3)}_{i})=4}$, it follows that $\tr{\mathcal{M}^{(3)}_i}=1$ for $i=1,2$ and $\tr{\mathcal{M}^{(3)}_3}=2$. By analogous reasoning, the $\mathcal{W}'_1$, and $\mathcal{W}'_2$ operators must be expressed as the same linear combination of Pauli matrices as stated in Theorem \ref{th1} and thus, we can say that the measurement operators $\mathcal{M}^{(3)}_1$ and $\mathcal{M}^{(3)}_2$ are local unitary equivalent to $\ket{\phi^+}\!\bra{\phi^+}$ and $\ket{\phi^-}\!\bra{\phi^-}$ respectively. This makes it clear that $\mathcal{M}^{(3)}_3$ must be local unitary equivalent to $\ket{\psi^+}\!\bra{\psi^+}+\ket{\psi^-}\!\bra{\psi^-}$. 
    
\textit{For robust self-testing:} the maximal quantum value of this communication task self-tests a partial Bell basis measurement for $k=3$, with two entangled elements and one product element. Thus, we focus on establishing the robustness of the entangled measurements operators’ fidelity rather than their average fidelity. Before delving into the robustness of $\mathcal{M}^{(3)}_1$ and $\mathcal{M}^{(3)}_2$, we show that $\mathcal{S}^{Comm}$ can be maximized only when the messages sent by each sender $j$, for each input $x_j$, correspond to antipodal Bloch vectors. Analogous to the discussion of the robustness of the GHZ measurements discussed in Sec. \ref{robustness}, we see that the Eq. \eqref{partial bell basis} can be reformulated as follows
\begin{equation} \label{decompo1}
    \begin{split}
    \mathcal{S}^{Comm}=\frac{1}{8\sqrt{2}}\tr& \Big[\Big(\mathcal{Q}_{00}+\mathcal{Q}_{10}\Big)f'_1(\mathcal{M})
    +\\&\Big(\mathcal{Q}_{01}-\mathcal{Q}_{11}\Big)f'_2(\mathcal{M})\Big],
    \end{split}
    \end{equation}
    where, $\mathcal{Q}_{x_1x_2}$ has the same expression as \eqref{qx1x2} and
    \begin{equation}
    \begin{split}
        &f'_1(\mathcal{M})=\mathcal{M}^{(3)}_1-\mathcal{M}^{(3)}_2,\\
        &f'_2(\mathcal{M})=\mathcal{M}^{(3)}_1+\mathcal{M}^{(3)}_2-2\mathcal{M}^{(3)}_3.
    \end{split}    
    \end{equation}
Following the same line of reasoning \eqref{tr_Q_f_M}-\eqref{ortho_cond},  we arrive at the same conclusion that the maximum quantum value of $\mathcal{S}^{Comm}$ can be attained only when, for each $x_j$, the messages sent by each sender $j$ must correspond to anti-podal Bloch vectors. As demonstrated in the proof of Theorem \ref{robust1}, we again emphasize that it is sufficient to consider the states being sent by the senders as antipodal. 

This allows us to parametrize the operators $\mathcal{A}_{x_j}^{(j)}$ as
\begin{equation}\label{param sin cos}
    \begin{split}
    &\mathcal{A}_{x_1}^{(1)}=\cos\alpha_1\sigma_X + (-1)^{x_1}\sin\alpha_1\sigma_Z,\\&
    \mathcal{A}_{x_2}^{(2)}=\cos\alpha_2\sigma_A + (-1)^{x_2}\sin\alpha_2\sigma_B,  
\end{split}    
\end{equation}
where ${\sigma_A} =(\sigma_X + \sigma_Z)/\sqrt{2} $,$\sigma_B=(\sigma_X - \sigma_Z)/\sqrt{2} $ and $\alpha_1,\alpha_2\in[0,\pi/2]$. The success metric $\mathcal{S}^{Comm}$ is a sum of three terms. For simplicity, we choose to lower bound each term by subtracting the maximum value that can be achieved by each term by the same value $\epsilon'$, namely the minimum of the three individual errors, and use this as a uniform error across all terms, i.e, 
\begin{equation}
    \begin{split}
       & \tr{(\mathcal{M}^{(3)}_1\mathcal{W'}_1)}\geq2\sqrt{2}-\epsilon',\\
       & \tr{(\mathcal{M}^{(3)}_2\mathcal{W'}_2)}\geq2\sqrt{2}-\epsilon',\\
       & \tr{(\mathcal{M}^{(3)}_3\mathcal{W'}_3)}\geq4\sqrt{2}-\epsilon'.
    \end{split}
\end{equation}
According to the Theorem \ref{robust1}, if the success metric $\mathcal{S}^{Comm}$is lower bounded by $1-\epsilon$, then it follows that $\epsilon'=(8\sqrt{2}\epsilon)/3$. The maximum eigenvalues of $\{\mathcal{W'}_i\}_i$'s, $\lVert(\mathcal{W'}_i)\rVert$ can be computed using \eqref{param sin cos} such that 
\begin{equation}
\begin{split}
    &\lVert\mathcal{W'}_i\rVert=2\sqrt{1+\sin2\alpha_1\sin2\alpha_2}\leq2\sqrt{2}, \quad i\in\{1,2\}\\
&\lVert\mathcal{W'}_3\rVert=4\sin\alpha_1\leq 4.
\end{split}    
\end{equation}
When ideally self-tested, $\lVert\mathcal{W'}_i\rVert=2\sqrt{2}$ for all $i\in\{1,2,3\}$ which implies that $\alpha_1=\alpha_2=\pi/4$. Using the trace condition $\tr(\mathcal{M}^{(3)}_i\mathcal{W'}_i)\leq\lVert\mathcal{W'}_i\rVert\tr({\mathcal{M}^{(3)}_i})$, we can get a lower bound of $\tr({\mathcal{M}^{(3)}_i})$ as  
\begin{equation}
    \begin{split}
        & \tr{(\mathcal{M}^{(3)}_1)}\geq1-\frac{\epsilon'}{2\sqrt{2}},\\
       & \tr{(\mathcal{M}^{(3)}_2)}\geq1-\frac{\epsilon'}{2\sqrt{2}},\\
       & \tr{(\mathcal{M}^{(3)}_3)}\geq\frac{4\sqrt{2}-\epsilon'}{4\sin\alpha_1}.
    \end{split}
\end{equation}
Since $\{\mathcal{M}^{(3)}_i\}_i$ constitutes a valid POVM on a four-dimensional Hilbert space, the sum of their traces must satisfy $\tr{(\sum_i\mathcal{M}^{(3)}_i)=4}$. Therefore, knowing the minimum values of $\tr{(\mathcal{M}^{(3)}_2)}$, and $\tr{(\mathcal{M}^{(3)}_3)}$ and Taylor expanding $\csc \alpha_1$ around $\pi/4$ upto the linear order allows us to determine the maximum value that $\tr{(\mathcal{M}^{(3)}_1)}$ can achieve. A similar procedure can be followed for $\tr{(\mathcal{M}^{(3)}_2)}$ such that
\begin{equation}\label{trace of measurement equation}
    \tr{(\mathcal{M}^{(3)}_i)}\leq 1+\frac{\epsilon'}{\sqrt{2}}+\bigg(2-\frac{\epsilon'}{2\sqrt{2}}\bigg)\Delta\alpha_1 \quad i\in\{1,2\},
\end{equation}
where $\Delta \alpha_1=\alpha_1-\pi/4$. If the allowed truncation error while expanding $\csc\alpha_1$ is $e$, then for \eqref{trace of measurement equation} to faithfully hold, $|\Delta \alpha_1|\leq\sqrt{4e/3\sqrt{2}}$. However, the fidelity bound must involve quantities that are determined by the success metrics. Thus, we must bound $\Delta \alpha_1$ with some function of $\mathcal{S}^{RAC}$, and it is easily achievable since by using \eqref{param sin cos} we can rewrite \eqref{srac} as
\begin{equation}
     \mathcal{S}^{RAC}\leq \frac{1}{2}+\frac{1}{8\sqrt{2}}[\sqrt{4+4\cos 2\alpha_1}+\sqrt{4-4\cos 2\alpha_1}],
\end{equation}
which leads to $\Delta\alpha_1$ being bounded as
\begin{equation}
    \Delta\alpha_1\leq \cos^{-1}\bigg(2\sqrt{2}\bigg(\mathcal{S}^{RAC}-\frac{1}{2}\bigg)\bigg).
\end{equation}
One can always choose the truncation error such that
\begin{equation}
    \sqrt{\frac{4e}{3\sqrt{2}}}\approx \cos^{-1}\bigg(2\sqrt{2}\bigg(\mathcal{S}^{RAC}-\frac{1}{2}\bigg)\bigg).
\end{equation}
We can then reformulate Eq. \eqref{trace of measurement equation} such that we arrive at
\begin{equation}\label{trace inequality1}
\begin{split}
    \tr{(\mathcal{M}^{(3)}_i)}&\leq 1+\frac{\epsilon'}{\sqrt{2}}\\
    &+\bigg(2-\frac{\epsilon'}{2\sqrt{2}}\bigg) \cos^{-1}\bigg(2\sqrt{2}\bigg(\mathcal{S}^{RAC}-\frac{1}
    {2}\bigg)\bigg).
\end{split}
\end{equation}

We have already discussed the operator inequalities used to quantify the fidelity between ideal and realized measurements in the proof of Theorem \ref{robust1}. Consequently, we have the operator inequalities
\begin{subequations}
    \begin{align}
        &\Lambda(\ket{\phi^+}\!\bra{\phi^+})\geq r\mathcal{W'}_1+\mu\mathbb{1}\label{3 outcome proj},\\&
        \Lambda(\ket{\phi^-}\!\bra{\phi^-})\geq r\mathcal{W'}_2+\mu\mathbb{1}\label{3 outcome proj1},
    \end{align}
\end{subequations}
with $r$ and $\mu$ equal to those chosen in Theorem \ref{robust1} for the two-sender, one-receiver case.
Multiplying \eqref{3 outcome proj} and \eqref{3 outcome proj1} with $\mathcal{M}^{(3)}_1$ and $\mathcal{M}^{(3)}_2$ and tracing out, we come up with the following inequalities for $i\in\{1,2\}$
\begin{equation}\label{3 outcome proj2}
    {F}(\mathcal{M}^{(3)}_i)\geq r\tr{(\mathcal{M}^{(3)}_i\mathcal{W'}_i)}+\mu\tr({\mathcal{M}^{(3)}_i}).
\end{equation}
Since the value of $\mu$ is negative, we can place \eqref{trace inequality1} in \eqref{3 outcome proj2}. Finally, making use of the fact that $2\sqrt{2}r+\mu=1$, we arrive at \eqref{fidelity partial}.

This completes the proof of the ideal and robust self-testing of the partial Bell basis measurement.
\leavevmode
\end{proof}

As stated in the discussion following the proof of Theorem \ref{robust1}, the fidelity bound must be greater than $1/2$ for a verifier to draw an inference about the degree of entanglement in the realized measurement operator. 
However, a word of caution is in order. \re A bit of careful observation reveals that the fidelity bound \eqref{fidelity partial} approaches $1$ when both $\mathcal{S}^{Comm}\rightarrow1$ and $\mathcal{S}^{RAC}\rightarrow (1+1/\sqrt{2})/2$. In other words, if one has $\mathcal{S}^{RAC}<(1+1/\sqrt{2})/2$ then the fidelity bound is \textit{not} tight. Thus, for cases where we have a non-optimal value of $\mathcal{S}^{RAC}$, requiring the fidelity bound in \eqref{fidelity partial} to be greater than or equal to $1/2$ implies that we need
\begin{equation} \label{eps_bound_partial}
    \begin{split}
        &\epsilon\leq f_1(\mathcal{S}^{RAC})/f_2(\mathcal{S}^{RAC}) \textrm{ where,}\\
        &f_1(\mathcal{S}^{RAC})=1/2+2\mu\cos^{-1}\left(2\sqrt{2}\left(\mathcal{S}^{RAC}-1/2\right)\right) \textrm{ and,}\\
        &f_2(\mathcal{S}^{RAC})=(8\sqrt{2}/3)(r-\mu/\sqrt{2})+\\
        &(4\mu/3)\cos^{-1}\left(2\sqrt{2}\left(\mathcal{S}^{RAC}-1/2\right)\right).
    \end{split}
\end{equation}
The upper bound on the error is meaningful only when it is strictly greater than $0$, which necessitates that both $f_1(\mathcal{S}^{RAC})>0$ and $f_2(\mathcal{S}^{RAC})>0$. This implies that one must always have
\begin{equation}
    \begin{split}
        \mathcal{S}^{RAC}>&\frac{1}{2}+\frac{1}{2\sqrt{2}}\cos\left(\frac{1}{4|\mu|}\right)\\
        =&\frac{1}{2}+\frac{1}{2\sqrt{2}}\cos\left(\frac{1}{1+2\sqrt{2}}\right)\\
        \approx & 0.842.
    \end{split}
\end{equation}
 When $\mathcal{S}^{RAC}=(1+1/\sqrt{2})/2$, \eqref{eps_bound_partial} results in
\begin{equation} \label{eps_upper}
    \epsilon\leq \frac{3}{14+12\sqrt{2}}\approx0.09.
\end{equation} \blk          
\begin{figure}[h] 
    \centering
    \includegraphics[width=9cm, height=5cm]{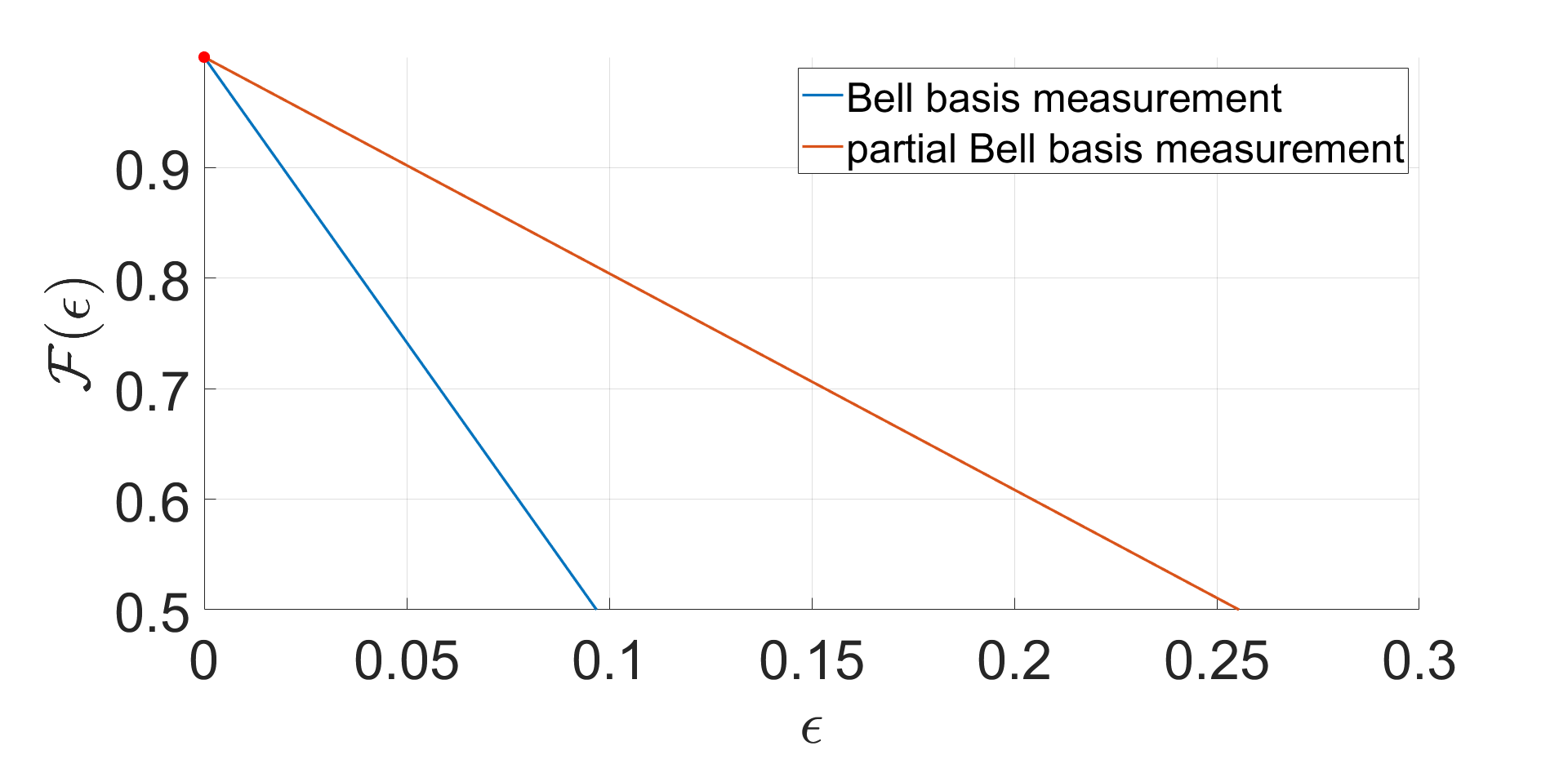}
    \caption{\re Plot illustrating the relationship between the deviation $\epsilon $ and achievable fidelity $\mathcal{F}(\epsilon)$ as a function of the deviation $\epsilon$ for two different scenarios. The red line indicates the robust self-testing of Bell basis measurement in a two sender-one receiver scenario where $\mathcal{S}\geq 1-\epsilon$. The blue line indicates the scenario when partial Bell basis measurement is robustly self-tested and $\mathcal{S}^{Comm}\geq 1-\epsilon$. Also, note that in this scenario, the graph shows the variation when $\mathcal{S}^{RAC}$ achieves its maximum value.\blk }
    \label{partial bell basis figuire}
\end{figure}

\section{Conclusion}\label{conclusion}
In this paper, we demonstrate a semi-device-independent self-testing protocol for measurements in an entangled basis, without requiring shared entanglement between parties. The approach is based on a communication task involving $n$ senders and a receiver, with success quantified by a specific metric. Theorem \ref{th1} showed that achieving the maximum quantum value of this metric requires the senders to transmit pure states and the receiver to perform entangled basis measurements. However, we explicitly showed an example where achieving this maximum alone does not always guarantee self-testability of the measurements. Using operator inequalities, the robustness of the measurement process under noise was studied. We also discussed the relevance of a communication scenario using partial Bell basis measurement, which can be realized in a laboratory using linear optics.

In future work it would be compelling to explore the self-testing of other classes of entangled measurements, particularly in higher-dimensional systems. In this work we chose to delve into a simpler scenario in which each sender receives four possible inputs. \re One may also consider the self-testing of other Bell inequalities, such as the Mermin inequality \cite{Kan}, to develop robust self-testing protocols for entangled basis measurements. \blk However, it would also be interesting to explore the case where each sender receives only three inputs. While recent advances have demonstrated the self-testing of nonprojective qubit measurements \cite{Tavakoli_2020, PhysRevA.110.032427} and unsharp measurements \cite{PhysRevResearch.2.033014} within prepare-and-measure scenarios, such investigations have yet to be extended to communication-based frameworks. Notably, prepare-and-measure approaches typically rely on assumptions about the underlying Hilbert space dimension. However, alternative physical constraints, such as bounds on the energy, entropy, or purity of the prepared states, may also serve as viable resources for enabling semi-device-independent self-testing.

\subsection*{Acknowledgment} We acknowledge financial support from STARS (STARS/STARS-2/2023-0809), Govt. of India.

\bibliography{ref}
\vspace{0.5 cm}
\appendix
\makeatletter
\renewcommand\section{\@startsection{section}{1}{0pt}
  {1.5ex plus 1ex minus .2ex}
  {1ex plus .2ex}
  {\centering\normalfont\bfseries}}
\makeatother
\renewcommand{\thesection}{\Alph{section}}
\renewcommand\appendixname{Appendix}
\onecolumngrid
\renewcommand\appendixname{Appendix}
\renewcommand{\thesection}{\Alph{section}}
\re \section{Explicit matrix representation of states and measurements to demonstrate that optimal quantum advantage does not always imply self-testing}

Consider the two-sender one-receiver communication task and the success metric discussed in Sec. \ref{no_advantage}. The maximum value that the success metric $\mathcal{S}$ can achieve, as defined in \eqref{S_no_self_test} is $\mathcal{S}\approx2.8284$. In Sec. \ref{no_advantage}, the entangling and the non-entangling measurements that result in optimal advantage and the corresponding message states were laid out. Here, we provide explicit matrix representations of those measurements and states (up to local unitaries), along with the lowest eigenvalue of the partially transposed measurements to check entanglement. Note that the matrices are expressed in the computational basis

    \begin{equation}
        \mathcal{M}^{\textrm{ent}}_0=\begin{bmatrix}
            0.1522 & -0.0408+i0.0205 & -0.0226+i0.0052 & 0.3550-i0.0189\\
            -0.0408-i0.0205 & 0.8507 & 0.3252-i0.1365 & 0.0228-i0.0050\\
            -0.0226-i0.0052 & 0.3252+i0.1365 & 0.1747 & -0.0138+i0.0249\\
            0.3550+i0.0189 & 0.0228+i0.0050 & -0.0138-i0.0249 & 0.8496
        \end{bmatrix}
    \end{equation}
    After partially transposing the matrix with respect to the second subsystem, the lowest eigenvalue we obtain is $-0.0021$, which implies that the measurement is entangled. The density matrices corresponding to the messages sent by the two senders $A^{(1)}$ and $A^{(2)}$ are 

    \begin{equation}
        \begin{aligned}
            \rho^{(1)}_1=\begin{bmatrix}
                1 & 0\\
                0 & 0
            \end{bmatrix}, && \rho^{(1)}_2=\begin{bmatrix}
                0.4783 & 0.4869-i0.1118\\
                0.4869+i0.1118 & 0.5217
            \end{bmatrix}, && \rho^{(1)}_3=\begin{bmatrix}
                0.4807 & 0.4870-i0.1116\\
                0.4870+i0.1116 & 0.5193
            \end{bmatrix}
        \end{aligned}
    \end{equation}
\begin{equation}
        \begin{aligned}
            \rho^{(2)}_1=\begin{bmatrix}
                1 & 0\\
                0 & 0
            \end{bmatrix}, && \rho^{(2)}_2=\begin{bmatrix}
                0.1265 & 0.3218+i0.0834\\
                0.3218-i0.0834 & 0.8735
            \end{bmatrix}, && \rho^{(2)}_3=\begin{bmatrix}
                0.5328 & -0.4916-i0.0854\\
                -0.4916+i0.0854 & 0.4672
            \end{bmatrix}
        \end{aligned}
    \end{equation}

    The nonentangled measurement is 

    \begin{equation}
        \mathcal{M}^{\textrm{non-ent}}_0=\begin{bmatrix}
            0.0836 & 0.2236+i0.0848 & -0.0119+i0.0522 &0.1244-i0.0327\\
            0.2236-i0.0848 & 0.9164 & -0.0979-i0.0835 & 0.0119-i0.0522\\
            -0.0119-i0.0522 & -0.0979+i0.0835 & 0.3331 & -0.0545-i0.4470\\
            0.1244+i0.0327 & 0.0119+i0.0522 & -0.0545+i0.4470 & 0.6669
            
        \end{bmatrix}.
    \end{equation}
    The lowest eigenvalue of the partial transposition of this measurement with respect to the second subsystem results in $0$, implying that the measurement is separable. The messages are
\begin{equation}
        \begin{aligned}
            \rho^{(1)}_1=\begin{bmatrix}
                1 & 0\\
                0 & 0
            \end{bmatrix}, && \rho^{(1)}_2=\begin{bmatrix}
                0.0405 & -0.0439+i0.1921\\
                -0.0439-i0.1921 & 0.9595
            \end{bmatrix}, && \rho^{(1)}_3=\begin{bmatrix}
                0.0405 & -0.0439+i0.1921\\
                -0.0439-i0.1921 & 0.9595
            \end{bmatrix}
        \end{aligned},
    \end{equation}
    
    \begin{equation}
        \begin{aligned}
            \rho^{(2)}_1=\begin{bmatrix}
                1 & 0\\
                0 & 0
            \end{bmatrix}, && \rho^{(2)}_2=\begin{bmatrix}
                0.3441 & -0.0667-i0.4704\\
                -0.0667+i0.4704 & 0.6559
            \end{bmatrix}, && \rho^{(2)}_3=\begin{bmatrix}
                0.3080 & 0.2140+i0.4091\\
                0.2140-i0.4091 & 0.6920
            \end{bmatrix}
        \end{aligned}.
    \end{equation}

\blk
\section{Sum-Of-Squares decomposition of $\mathcal{W}_s$ operators}\label{sos decomposition}

\begin{lemma} \label{lem1}
Let the maximum eigenvalue of an operator $\mathcal{O}$ be denoted by $\lVert \mathcal{O} \rVert$. Then, for the operators $\mathcal{W}_s$ as defined in \eqref{w_i_def}, for the case where $\mathcal{A}_{x_j}^{(j)}=\psi_{0|x_j}^{(j)}-\psi_{1|x_j}^{(j)}$, with $\psi_{a_j|x_j}^{(j)}=\ket{\psi_{a_j|x_j}^{(j)}}\!\bra{\psi_{a_j|x_j}^{(j)}}$, the following inequality holds for all $s$
    \begin{equation}
        \lVert \mathcal{W}_s \rVert \leq 2\sqrt{2}(n-1).
    \end{equation}
\end{lemma}
\begin{proof}
\vspace{-0.2 cm}
Let us consider the  operators $\{\mathcal{P}_{j,s}\}_{j=1}^n$, defined as
    \begin{equation}\label{operators}
    \mathcal{P}_{1,s}=(-1)^{s_1}\frac{1}{\sqrt{2}}\left(\mathcal{A}_0^{(1)}+\mathcal{A}_1^{(1)}\right)\otimes \bigotimes_{j=2}^n\mathcal{A}_0^{(j)},
        \mathcal{P}_{j,s}=(-1)^{s_j}\frac{1}{\sqrt{2}}\left(\mathcal{A}_0^{(1)}-\mathcal{A}_1^{(1)}\right)\otimes\mathcal{A}_1^{(j)} \text{  $\forall$ $j\in \{2,\cdots,n\}$},
    \end{equation}
  
which can be used to write a sum-of-squares (SOS) decomposition, $\beta_Q\mathbbm{1}-\mathcal{W}_s$, as defined in \eqref{shiifted bell operator}. In the following paragraphs, we show that $\beta_Q =2\sqrt{2}(n-1)$.
\vspace{-0.01 cm}
First, observe that
\begin{equation}\label{eqb3}
        \frac{n-1}{\sqrt{2}}\left(\mathbbm{1}-\mathcal{P}_{1,s}\right)^2+ \frac{1}{\sqrt{2}}\sum_{j=2}^n\left(\mathbbm{1}-\mathcal{P}_{j,s}\right)^2=\sqrt{2}(n-1)\mathbbm{1}-\mathcal{W}_s+\frac{n-1}{\sqrt{2}}\mathcal{P}_{1,s}^2+\frac{1}{\sqrt{2}}\sum_{j=2}^n\mathcal{P}_{j,s}^2.
\end{equation}

Rearranging the above equation makes it clear that $\sqrt{2}(n-1)\mathbbm{1}-\mathcal{W}_s$ cannot be written as an SOS because all the terms are not semidefinite positive. Thus, $\beta_Q\neq\sqrt{2}(n-1)$. However, by adding another $\sqrt{2}(n-1)\mathbbm{1}$ to the Eq. \eqref{eqb3}, we can write
    \begin{equation}\label{SOS}
        2\sqrt{2}\left(n-1\right)\mathbbm{1}-\mathcal{W}_s=\frac{n-1}{\sqrt{2}}\left(\mathbbm{1}-\mathcal{P}_{1,s}\right)^2+\frac{1}{\sqrt{2}}\sum_{j=2}^n\left(\mathbbm{1}-\mathcal{P}_{j,s}\right)^2+\sqrt{2}(n-1)\Bigg(\mathbbm{1}-\frac{1}{2(n-1)}\Bigg((n-1)\mathcal{P}_{1,s}^2+\sum_{j=2}^n\mathcal{P}_{j,s}^2\Bigg)\Bigg)
    \end{equation}

We now show that $\lVert(n-1)\mathcal{P}_{1,s}^2+\sum_{j=2}^n\mathcal{P}_{j,s}^2 \rVert \leq 2(n-1)$, which would imply that \eqref{SOS} is a valid SOS decomposition of the shifted $\mathcal{W}_i$ operator and consequently, $\beta_Q=2\sqrt{2}(n-1)$.
Note that the eigenvalues of $ \mathcal{A}_{x_j}^{(j)}=\pm\sqrt{1-|\la\psi_{0|x_j}^{(j)}|\psi_{1|x_j}^{(j)}\ra|^2}$, for all $j\in\{1,\cdots,n\}$ and $x_j\in\{0,1\}$ which implies that $-\mathbbm{1}\leq \mathcal{A}_{x_j}^{(j)} \leq \mathbbm{1}$ and consequently, $\left( \mathcal{A}_{x_j}^{(j)}\right)^2\leq \mathbbm{1}.$.
 With help of the operators $\mathcal{P}_1$ and $\mathcal{P}_j$ as defined in \eqref{operators} we can write the expression $ (n-1)\mathcal{P}_{1,s}^2+\sum_{j=2}^n\mathcal{P}_{j,s}^2=$ as

\begin{equation}
        (n-1)\mathcal{P}_{1,s}^2+\sum_{j=2}^n\mathcal{P}_{j,s}^2=\\
            \frac{n-1}{2}\left(\mathcal{A}_0^{(1)}+\mathcal{A}_1^{(1)}\right)^2\otimes \bigotimes_{j=2}^n\left(\mathcal{A}_0^{(j)}\right)^2+\\
            \frac{1}{2}\sum_{j=2}^n\left(\mathcal{A}_0^{(1)}-\mathcal{A}_1^{(1)}\right)^2\otimes\left(\mathcal{A}_1^{(j)}\right)^2,
\end{equation}
which, \re using the fact that $\left( \mathcal{A}_{x_j}^{(j)}\right)^2\leq \mathbbm{1}$, \blk  can be shown to be bounded as
    
\begin{equation}\label{equation A7}
     (n-1)\mathcal{P}_{1,s}^2+\sum_{j=2}^n\mathcal{P}_{j,s}^2 \leq \\\left(\left(\mathcal{A}_0^{(1)}+\mathcal{A}_1^{(1)}\right)^2+\left(\mathcal{A}_0^{(1)}-\mathcal{A}_1^{(1)}\right)^2\right)\otimes\mathbbm{1}_{2^{n-1}}.
\end{equation}
\re Upon expanding the right hand side of \eqref{equation A7} and canceling out the cross terms $\mathcal{A}_0^{(1)}\mathcal{A}_1^{(1)}$ and $\mathcal{A}_1^{(1)}\mathcal{A}_0^{(1)}$,Eq. \eqref{equation A7} reduces to
\begin{equation}
    (n-1)\mathcal{P}_{1,s}^2+\sum_{j=2}^n\mathcal{P}_{j,s}^2 \leq 
    ({n-1})\left(\left(\mathcal{A}_0^{(1)}\right)^2+\left(\mathcal{A}_1^{(1)}\right)^2\right)\otimes\mathbbm{1}_{2^{n-1}}.
\end{equation}\blk
\re Again using $\left( \mathcal{A}_{x_j}^{(j)}\right)^2\leq \mathbbm{1}$, \blk, it follows that $\lVert(n-1)\mathcal{P}_{1,s}^2+\sum_{j=2}^n\mathcal{P}_{j,s}^2 \rVert \leq 2(n-1)$, which implies that $2\sqrt{2}(n-1)\mathbbm{1}-\mathcal{W}_s\geq0$, or in other words, the maximum eigenvalue of $\mathcal{W}_s$ is upper-bounded by $2\sqrt{2}(n-1)$.
\end{proof}
However, the qubit states sent by the senders need not be pure. We show in the next lemma that even when one takes into consideration mixed states, the maximum eigenvalue of $\mathcal{W}_s$ cannot exceed $2\sqrt{2}(n-1)$. 
\begin{lemma} \label{lemm2}
    The maximum eigenvalue of the operators $\mathcal{W}_s$ as defined in \eqref{w_i_def} are upper-bounded by $2\sqrt{2}(n-1)$, even when $\mathcal{A}_{x_j}^{(j)}=\rho_{0|x_j}^{(j)}-\rho_{1|x_j}^{(j)}$, where $\rho_{a_j|x_j}^{(j)}$ are qubit mixed states.
\end{lemma}
\begin{proof}
    To begin with, let us fix $j$ for a particular sender and examine a specific operator  $\mathcal{A}_{0}^{(j)}$ which is equal to $\rho_{0|0}^{(j)}-\rho_{1|0}^{(j)}$. The mixed states can be spectrally decomposed such that, $\rho_{0|0}^{(j)}=p_1^{(j)}\psi_{0|0}^{(j)}+(1-p_1^{(j)}) \overline{\psi}_{0|0}^{(j)}$ and $\rho_{1|0}^{(j)}=p_2^{(j)}\psi_{1|0}^{(j)}+(1-p_2^{(j)}) \overline{\psi}_{1|0}^{(j)}$. Observe that the operator $\mathcal{A}_{0}^{(j)}$ can be written as
    \begin{equation} \label{A_convex_mix}
\mathcal{A}_{0}^{(j)}=p_1^{(j)}p_2^{(j)}\left(\psi_{0|0}^{(j)}-\psi_{1|0}^{(j)}\right)
    +p_1^{(j)}(1-p_2^{(j)})\left(\psi_{0|0}^{(j)}-\overline{\psi}_{1|0}^{(j)}\right)
    +(1-p_1^{(j)})p_2^{(j)}\left(\overline{\psi}_{0|0}^{(j)}-\psi_{1|0}^{(j)}\right)
    +(1-p_1^{(j)})(1-p_2^{(j)})\left(\overline{\psi}_{0|0}^{(j)}-\overline{\psi}_{1|0}^{(j)}\right).
    \end{equation}
    Careful probing shows that the sum of the coefficients in this equation adds up to one. Also, observe that each term such as $\psi_{0|0}^{(j)}-\psi_{1|0}^{(j)}$ or $\psi_{0|0}^{(j)}-\overline{\psi}_{1|0}^{(j)}$ is an equivalent of $A_0^{(j)}$ where instead of being the difference of mixed states, they are the difference of pure states. In other words, $\mathcal{A}_{0}^{(j)}$ can be written as a convex mixture of $\tilde{\mathcal{A}}_{0,l}^{(j)}$, that is, $\mathcal{A}_{0}^{(j)}=\sum_{l=1}^4\tilde{p}_l^{(j)}\tilde{\mathcal{A}}_{0,l}^{(j)}$ where $\tilde{p}_1^{(j)}=p_1^{(j)}p_2^{(j)}$ and $\tilde{\mathcal{A}}_{0,1}^{(j)}=\psi_{0|0}^{(j)}-\psi_{1|0}^{(j)}$ and similarly $\tilde{p}_{l}^{(j)}$ and $\tilde{\mathcal{A}}_{0,l}^{(j)}$ for $l=2,3,4$ are defined according to the remaining terms in \eqref{A_convex_mix}. It is easy to see that one can follow a similar method and write $\mathcal{A}_1^{(j)}=\sum_{l=1}^4\tilde{q}_l^{(j)}\tilde{\mathcal{A}}_{1,l}^{(j)}$. Plugging these in \eqref{w_i_def} we see that
    \begin{equation} \label{w_i_convex1}
    \begin{split}
    \mathcal{W}_s=
            &(n-1)(-1)^{s_1}\left(\sum_{l_1=1}^4\tilde{p}_{l_1}^{(1)}\tilde{\mathcal{A}}_{0,l_1}^{(1)}+\sum_{m_1=1}^4\tilde{q}_{m_1}^{(1)}\tilde{\mathcal{A}}_{1,m_1}^{(1)}\right) \bigotimes_{j=2}^n\sum_{l_j=1}^4\tilde{p}_{l_j}^{(j)}\tilde{\mathcal{A}}_{0,l_j}^{(j)}\\
            &+\sum_{j=2}^n(-1)^{s_j}\left(\sum_{l_1=1}^4\tilde{p}_{l_1}^{(1)}\tilde{\mathcal{A}}_{0,l_1}^{(1)}-\sum_{m_1=1}^4\tilde{q}_{m_1}^{(1)}\tilde{\mathcal{A}}_{1,m_1}^{(1)}\right)\otimes\sum_{m_j=1}^4\tilde{q}_{m_j}^{(j)}\tilde{\mathcal{A}}_{1,m_j}^{(j)}.
            \end{split}
         \end{equation}
    Note that no term in \eqref{w_i_convex1} contains the product of every possible convex coefficient. However, we safely insert the missing sums of coefficients since they add up to one, and pull out the products of all the coefficients \vspace{-0.5 cm}such that we end up with
    \begin{equation} \label{w_i_convex2}
        \mathcal{W}_s=\sum_{l_1,\cdots,l_n;m_1,\cdots,m_n=1}^4\prod_{j,k=1}^n\left(\tilde{p}_{l_j}^{(j)}\tilde{q}_{m_k}^{(k)}\right)\Tilde{\mathcal{W}}_s^{l_1,\cdots,l_n;m_1,\cdots,m_n},
    \end{equation}
    where $\Tilde{\mathcal{W}}_s^{l_1,\cdots,l_n;m_1,\cdots,m_n}$ is defined in accordance with \eqref{w_i_def} using $\tilde{\mathcal{A}}_{0,l_j}^{(j)}$ and $\tilde{\mathcal{A}}_{1,m_k}^{(k)}$ for different index values of $l_j$ and $m_k$. From Lemma \ref{lem1} we know that $\lVert\Tilde{\mathcal{W}}_s^{l_1,\cdots,l_n;m_1,\cdots,m_n}\rVert\leq 2\sqrt{2}(n-1)$, and since \eqref{w_i_convex2} represents $\mathcal{W}_s$ as a convex mixture of $\Tilde{\mathcal{W}}_s^{l_1,\cdots,l_n;m_1,\cdots,m_n}$, it follows that $\lVert \mathcal{W}_s \rVert \leq 2\sqrt{2}(n-1)$.Now, $\lVert\Tilde{\mathcal{W}}_s^{l_1,\cdots,l_n;m_1,\cdots,m_n}\rVert=2\sqrt{2}(n-1)$ for all values of $\{l_1,\cdots,l_n;m_1,\cdots,m_n\}$ further implies that for all possible $l_j$ and $m_k$ $\lVert \tilde{\mathcal{A}}_{0,l_j}^{(j)}\rVert =1$ and $\lVert \tilde{\mathcal{A}}_{1,m_k}^{(k)}\rVert =1$. For example, let us consider $\tilde{\mathcal{A}}_{0,1}^{(k)}=\psi_{0|0}^{(k)}-\psi_{1|0}^{(k)}$. We see that $\lVert \tilde{\mathcal{A}}_{0,1}^{(k)} \rVert=1$ demands that $\la \psi_{0|0}^{(k)}|\psi_{1|0}^{(k)}\ra=0$, and similarly $\lVert \tilde{\mathcal{A}}_{0,l}^{(k)} \rVert=1$ for $l=2,3,4$ demands that $\la \psi_{0|0}^{(k)}|\overline{\psi}_{1|0}^{(k)}\ra=0$, $\la \overline{\psi}_{0|0}^{(k)}|\psi_{1|0}^{(k)}\ra=0$ and $\la \overline{\psi}_{0|0}^{(k)}|\overline{\psi}_{1|0}^{(k)}\ra=0$.  
\end{proof}

\section{Bloch vector antipodality in $n$-Sender, single-receiver configurations}\label{structure_of_metric}
\begin{lemma}
    For any given measurement $\{\mathcal{M}_s\}_s$, the success metric $\mathcal{S}$ as defined in \eqref{SGen} is maximized for a set of messages which are pure and antipodal, that is,
\begin{equation}
        \begin{split}    \rho_{a_j|x_j}^{(j)}=\ket{\psi_{a_j|x_j}^{(j)}}\!\bra{\psi_{a_j|x_j}^{(j)}}, \quad \textrm{and } \quad \bra{\psi_{0|x_j}^{(j)}}\psi_{1|x_j}^{(j)}\rangle=0, \quad \forall j\in\{1,2,\cdots,n\} \textrm{ and } x_j\in\{0,1\}.
        \end{split}
\end{equation}
\end{lemma}
\begin{proof}
Consider the success metric given by \eqref{SGen}. Each $W_s=\tr(\mathcal{W}_s\mathcal{M}_s)$ where $\mathcal{W}_s$ is defined in terms of $\mathcal{A}_{x_k}^{(k)}=\rho_{0|x_k}^{(k)}-\rho_{1|x_k}^{(k)}$ according to \eqref{w_i_def}. For any sender $A^{(k)}$, let $\mathcal{A}_{0}^{(k)}=\sum_{a_k=0,1}(-1)^{a_k}\rho_{a_k|0}^{(k)}$ and $\mathcal{A}_{1}^{(k)}=\sum_{b_k=0,1}(-1)^{b_k}\rho_{b _k|1}^{(k)}$. Expressed in this fashion, \eqref{w_i_def} can be rewritten as

\begin{equation}
\begin{split}
    \mathcal{W}_s=&(n-1)(-1)^{s_1}\left(\sum_{a_1=0,1}(-1)^{a_1}\rho_{a_1|0}^{(1)}+\sum_{b_1=0,1}(-1)^{b_1}\rho_{b_1|1}^{(1)}\right)\otimes\bigotimes_{j=2}^n\left(\sum_{a_j=0,1}(-1)^{a_j}\rho_{a_j|0}^{(j)}\right)+\\&\sum_{j=2}^n(-1)^{s_j}\left(\sum_{a_1=0,1}(-1)^{a_1}\rho_{a_1|0}^{(1)}-\sum_{b_1=0,1}(-1)^{b_1}\rho_{b_1|1}^{(1)}\right)\otimes\left(\sum_{b_j=0,1}(-1)^{b_j}\rho_{b_j|1}^{(j)}\right).
    \end{split}
\end{equation}
Substituting this expression of $\mathcal{W}_s$ in \eqref{SGen}, we obtain
\begin{equation}
    \begin{split}
        \mathcal{S}=\frac{1}{2^n(n-1)2\sqrt{2}}\tr\left((n-1)\mathcal{Q}_1f_1(\mathcal{M})+\sum_{j=2}^n\mathcal{Q}_jf_j(\mathcal{M})\right),
    \end{split}
\end{equation}
where $\{f_j(\mathcal{M})\}_{j=1}^n$ are defined as
\begin{equation}
    f_j(\mathcal{M})=\sum_s(-1)^{s_j}\mathcal{M}_s \textrm{ } \forall \textrm{ } j\in\{1,\cdots,n\}.
\end{equation}
The quantity $\mathcal{Q}_1$ is defined as
\begin{equation}
    \begin{split}
        &\mathcal{Q}_1=\mathcal{Q}_{11}+\mathcal{Q}_{12}\textrm{ } \textrm{ where,}\\
        & \mathcal{Q}_{11}=\sum_{a_1,a_2,\cdots,a_n=0,1}(-1)^{\bigoplus_{j=1}^n a_j}\bigotimes_{k=1}^n\rho_{a_k|0}^{(k)}\textrm{ } \textrm{ and,}\\
        &\mathcal{Q}_{12}=\sum_{b_1,a_2,\cdots,a_n=0,1}(-1)^{b_1 \oplus \bigoplus_{j=2}^n a_j}\rho_{b_1|1}^{(1)}\otimes\bigotimes_{k=2}^n\rho_{a_k|0}^{(k)},
    \end{split}
\end{equation}
while the rest of the $\{\mathcal{Q}_j\}_{j=2}^n$ are defined as
\vspace{-0.3 cm}
\begin{equation}
    \begin{split}
        &\mathcal{Q}_j=\mathcal{Q}_{j1}+\mathcal{Q}_{j2} \textrm{ } \textrm{ where,}\\
        &\mathcal{Q}_{j1}=\sum_{a_1,b_j=0,1}(-1)^{a_1 \oplus b_j}\rho_{a_1|0}^{(1)}\otimes\rho_{b_j|1}^{(j)}\textrm{ } \textrm{ and,}\\
        &\mathcal{Q}_{j2}=\sum_{b_1,b_j=0,1}(-1)^{b_1 \oplus b_j \oplus 1}\rho_{b_1|1}^{(1)}\otimes\rho_{b_j|1}^{(j)}.
    \end{split}
\end{equation}
Note that every $f_j(\mathcal{M})$ is a Hermitian operator since it is a linear combination of POVM elements and the structure of $\mathcal{Q}_j$ is such that one can write $\mathcal{S}$ as a sum of terms which are of the form $\tr(\rho f_j(\mathcal{M})-\overline{\rho}f_j(\mathcal{M}))$ where $\rho$ and $\overline{\rho}$ are density matrices. For example, consider the first set of summed terms in $\mathcal{Q}_1$, that is $\mathcal{Q}_{11}$. It is a sum of $2^n$ terms, and they can be rewritten as a sum of $2^{n-1}$ terms, where each such term is a difference of two density matrices $\rho$ and $\overline{\rho}$. In each such term, if $\rho$ has the form,
\begin{equation}
    \rho=\rho_{a_1|0}^{(1)}\otimes\rho_{a_2|0}^{(2)}\otimes\cdots\otimes\rho_{0|0}^{(k)}\otimes\cdots\otimes\rho_{a_n|0}^{(n)},
\end{equation}
then $\overline{\rho}$ has the form
\begin{equation}
    \overline{\rho}=\rho_{a_1|0}^{(1)}\otimes\rho_{a_2|0}^{(2)}\otimes\cdots\otimes\rho_{1|0}^{(k)}\otimes\cdots\otimes\rho_{a_n|0}^{(n)},
\end{equation}
such that when traced out after a matrix multiplication with $f_1(\mathcal{M})$, the term $\tr(\rho f_1(\mathcal{M})-\overline{\rho}f_1(\mathcal{M}))$ reads

\begin{equation}
\tr(\rho f_1(\mathcal{M})-\overline{\rho}f_1(\mathcal{M}))=\tr\Big[\Big(\rho_{a_1|0}^{(1)}\otimes\cdots\otimes\rho_{a_{k-1}|0}^{(k-1)}\otimes(\rho_{0|0}^{(k)}-\rho_{1|0}^{(k)})\otimes\cdots\otimes\rho_{a_n|0}^{(n)}\Big)f_1(\mathcal{M})\Big].
\end{equation}

The above term attains its maximum value when $\rho$ and $\overline{\rho}$ are rank-1 projectors and eigenstates of $f_1(\mathcal{M})$ such that $\rho$ corresponds to the maximum eigenvalue of $f_1(\mathcal{M})$ and $\overline{\rho}$ to its minimum eigenvalue. This implies that $\rho$ and $\overline{\rho}$ must be orthogonal, which in turn implies that $\rho_{0|0}^{(k)}$ and $\rho_{1|0}^{(k)}$ must be orthogonal. A similar conclusion can be arrived at by considering the terms in $\mathcal{Q}_{12}$ and the terms in $\mathcal{Q}_j$ for $j \in \{2,3,\cdots,n\}$, such that it turns out that to attain the maximum possible value of $\mathcal{S}$, each sender $A^{(j)}$ must send messages which are pure and for a given $x_j\in \{0,1\}$ the messages $\rho_{0|x_j}^{(j)}=\ket{\psi_{0|x_j}^{(j)}}\!\bra{\psi_{0|x_j}^{(j)}}$ and $\rho_{1|x_j}^{(j)}=\ket{\psi_{1|x_j}^{(j)}}\!\bra{\psi_{1|x_j}^{(j)}}$ must be orthogonal, that is, $\bra{\psi_{1|x_j}^{(j)}}\psi_{0|x_j}^{(j)}\ra=0.$
\end{proof}
\end{document}